\documentclass[a4paper, 11pt,reqno]{amsart}

\usepackage[table]{xcolor}

\usepackage[utf8]{inputenc}
\usepackage[T1]{fontenc}
\usepackage[colorlinks]{hyperref}
\usepackage{mathtools, amsthm, amsfonts, amssymb}
\usepackage{accents} 
\usepackage{microtype}

\hypersetup{
  linkcolor=[rgb]{0.3,0.3,0.6},
  citecolor=[rgb]{0.2, 0.6, 0.2},
  urlcolor=[rgb]{0.6, 0.2, 0.2}
}

\usepackage{latexsym}
\mathtoolsset{showonlyrefs,showmanualtags}
\usepackage{thmtools}
\usepackage{mathrsfs}
\usepackage{textcomp}
\usepackage{graphicx}
\usepackage{setspace}
\usepackage{nicefrac}
\usepackage{enumerate}
\usepackage{wasysym}
\usepackage[normalem]{ulem}
\usepackage{upgreek}

\usepackage{paralist}
\usepackage{etoolbox}
\usepackage{mathdots}
\usepackage{wrapfig}
\usepackage{floatflt}
\usepackage{tensor} 
\usepackage{lscape}
\usepackage{stmaryrd}
\usepackage[enableskew]{youngtab}
\usepackage{ytableau}
\usepackage{pifont}

\usepackage{tikz}
\usetikzlibrary{arrows}
\usetikzlibrary{matrix}
\usetikzlibrary{positioning}
\usetikzlibrary{calc}

\usepackage[all,cmtip]{xy}
\usepackage[labelfont=bf, font={small}]{caption}[2005/07/16]

\definecolor{red}{rgb}{1,0,0}

\newcommand{\xto}[1]{\xrightarrow{\phantom{a}{#1}{\phantom{a}}}}

\newcommand{\vvirg}{ , \dots , }

\newcommand{\contract}{\rotatebox[origin=c]{180}{ \reflectbox{$\neg$} }}

\newcommand{\bfB}{\mathbf{B}}

\newcommand{\bfT}{\mathbf{T}}
\newcommand{\bfU}{\mathbf{U}}
\newcommand{\bfV}{\mathbf{V}}
\newcommand{\bfW}{\mathbf{W}}
\newcommand{\bfX}{\mathbf{X}}
\newcommand{\bfY}{\mathbf{Y}}

\newcommand{\bfm}{\mathbf{m}}
\newcommand{\bfn}{\mathbf{n}}

\newcommand{\bfv}{\mathbf{v}}

\newcommand{\calK}{\mathcal{K}}

\newcommand{\calS}{\mathcal{S}}
\newcommand{\calT}{\mathcal{T}}

\newcommand{\bbC}{\mathbb{C}}

\newcommand{\bbN}{\mathbb{N}}

\newcommand{\bbQ}{\mathbb{Q}}

\newcommand{\bbS}{\mathbb{S}}

\renewcommand{\phi}{\varphi}

\renewcommand{\tilde}[1]{\widetilde{#1}}

\newcommand{\id}{\mathrm{id}}

\newcommand{\rank}{\mathrm{rank}}

\DeclareMathOperator{\sgn}{sgn}

\DeclareMathOperator{\depth}{depth}

\DeclareMathAccent{\wtilde}{\mathord}{largesymbols}{"65}

\newcommand{\Gr}{\mathrm{Gr}}

\newcommand{\GL}{\mathrm{GL}}
\newcommand{\SL}{\mathrm{SL}}

\newcommand{\calTNS}{\mathcal{T\!N\!S}}

\usepackage{hyperref}
\makeindex

\newcommand{\SI}{\mathrm{SI}}
\newcommand{\lcm}{\mathrm{lcm}}
\usetikzlibrary{decorations.markings}
\tikzset{->-/.style={decoration={
  markings,
  mark=at position .5 with {\arrow{angle 90}}},postaction={decorate}}}

\newtheorem{theorem}{Theorem}[section]
\newtheorem{lemma}[theorem]{Lemma}
\newtheorem{definition}[theorem]{Definition}
\newtheorem{proposition}[theorem]{Proposition}
\newtheorem{corollary}[theorem]{Corollary}
\newtheorem{conjecture}[theorem]{Conjecture}

\newcounter{claimCount}

\newcommand{\qmincut}{\mathrm{QMinCut}}
\newcommand{\qmaxflow}{\mathrm{QMaxFlow}}
\newcommand{\qcap}{\mathrm{qcap}}
\newcommand{\im}{\mathrm{im}}

\newcommand{\bridgegraphsmall}[5]{
	\begin{tikzpicture}[scale = .5,baseline={([yshift=-.5ex]current bounding box.center)}]
	\draw[thick] (0,0) --(.9,0);
	\draw[thick]  (0,1) --(.9,1);
	\draw[thick]  (1.1,0) --(2,0);
	\draw[thick]  (1.1,1) --(2,1);
	\draw[thick]  (1,0.1) --(1,.9);
	\node[scale = .7] at (.5,1.3) {$#1$};
	\node[scale = .7] at (1.5,1.3) {$#2$};
	\node[scale = .7] at (.5,-.3) {$#4$};
	\node[scale = .7] at (1.5,-.3) {$#5$};
	\node[scale = .7] at (1.3,.5) {$#3$};
\end{tikzpicture}
}

\newcommand{\bridgegraphsmallsub}[5]{
	\begin{tikzpicture}[scale = .5,baseline={([yshift=-.5ex]current bounding box.center)}]
	\draw[thick] (0,0) --(.9,0);
	\draw[thick]  (0,1) --(.9,1);
	\draw[thick]  (1.1,0) --(2,0);
	\draw[thick]  (1.1,1) --(2,1);
	\draw[thick]  (1,0.1) --(1,.9);
	\node[scale = .7] at (.5,1.3) {$#1$};
	\node[scale = .7] at (1.5,1.3) {$#2$};
	\node[scale = .7] at (.5,-.5) {$#4$};
	\node[scale = .7] at (1.5,-.5) {$#5$};
	\node[scale = .7] at (1.3,.5) {$#3$};
\end{tikzpicture}
}

\newcommand{\bridgegraphbig}[5]{
	\begin{tikzpicture}[scale = .5,baseline={([yshift=-.5ex]current bounding box.center)}]
	\draw[thick] (0,0) --(1.4,0);
	\draw[thick]  (0,1.3) --(1.4,1.3);
	\draw[thick]  (1.5,0) --(2.9,0);
	\draw[thick]  (1.5,1.3) --(2.9,1.3);
	\draw[thick]  (1.45,0.1) --(1.45,1.2);
	\node[scale = .6] at (.7,1.6) {$#1$};
	\node[scale = .6] at (2.25,1.6) {$#2$};
	\node[scale = .6] at (.7,-.3) {$#4$};
	\node[scale = .6] at (2.25,-.3) {$#5$};
	\node[scale = .6] at (1.7,.65) {$#3$};
\end{tikzpicture}
}

\usepackage{mathrsfs}

 \definecolor{darkspringgreen}{rgb}{0.09, 0.45, 0.27}

\title{Quantum max-flow in the bridge graph}
\author{Fulvio Gesmundo, Vladimir Lysikov, Vincent Steffan}

\address[F. Gesmundo]{Universität des Saarlandes, Saarbr\"ucken, Germany; (current) Institut de Mathématiques de Toulouse, Université Paul Sabatier, Toulouse, France}
\email[]{fgesmund@math.univ-toulouse.fr}

\address[V. Lysikov]{Ruhr-Universität Bochum, Germany}
\email[]{vladimir.lysikov@rub.de}

\address[V. Steffan]{QMATH, Dept. Math. Sci., U. of Copenhagen, Denmark}
\email[]{vincent.steffan@googlemail.com}

\makeatletter
\@namedef{subjclassname@2020}{%
  \textup{2020} Mathematics Subject Classification}
\makeatother

\usepackage[top=3.2cm,bottom=3.2cm,left=3cm,right=3cm]{geometry}

\begin{document}

\begin{abstract}
The quantum max-flow is a linear algebraic version of the classical max-flow of a graph, used in quantum many-body physics to quantify the maximal possible entanglement between two regions of a tensor network state. In this work, we calculate the quantum max-flow exactly in the case of the \emph{bridge graph}. The result is achieved by drawing connections to the theory of prehomogenous tensor spaces and the representation theory of quivers. Further, we highlight relations to invariant theory and to algebraic statistics.
\end{abstract}

\maketitle

\section{Introduction}
Tensor networks are a powerful tool to construct quantum many-body states~\cite{tns1,Orus_practical}. They define a class of tensors obtained via tensor contraction described combinatorially by the structure of a graph. They are used in quantum many-body physics to parameterize an ansatz class of quantum states with certain special entanglement properties, for example, states obeying an area law~\cite{Hastings_2007,RevModPhys.82.277}. Because of their desirable entanglement properties they have a wide range of applications such as holography~\cite{holography1,holography2,holography3,holography4,holography5}, quantum chemistry~\cite{chem1,chem2,chem3,chem4,chem5} and machine learning~\cite{ml1,ml2,ml3}. They also are powerful numerical tools \cite{Hauschild_2018,ChrGesStWer:Optimization} and key components of state of the art quantum circuit  simulators~\cite{simula1,simula2,simula3,sim1,sim2}.

Informally, tensor network states are quantum states built from ``smaller states'' via tensor contractions encoded in the combinatorics of a graph. We provide the precise definition in \autoref{sec:qmaxflowproblem}. For now, we mention that given a graph $\Gamma = (V,E)$ together with two assignments of integer weights $\bfm: E \rightarrow \mathbb{N}$ and $\textbf{n}:V\rightarrow \mathbb{N}$ on the sets of edges and vertices, respectively, one can define a set of unnormalized quantum states on $|V|$ parties, or equivalently, tensors of order $|V|$ which we denote by $\calTNS^\Gamma(\bfm,\bfn)$ in $\bigotimes_{v \in V} \bbC^{\bfn(v)}$, called the set of tensor network states associated to $\Gamma$ with \emph{bond dimensions} $\bfm$. 

Understanding properties of tensor network states from the point of view of geometry and information theory provides valuable insights on the physical states that they define. However, usually quantities associated to tensor network states are difficult to compute. In \cite{calegari2009positivity}, the quantum max-flow of a network was introduced, in the context of positivity of certain tensor operators. This is an information-theoretic version of the classical max-flow of a graph. The exact definition will be given in \autoref{sec:qmaxflowproblem}; informally, given a network $\Gamma = (V,E)$ with bond dimensions $\bfm$, and two disjoint subsets $\calS,\calT \subseteq V$ of the set of vertices of $\Gamma$, the associated quantum max-flow, denoted $\qmaxflow(\Gamma,\textbf{m},\mathcal{S},\mathcal{T})$, is the maximum possible Schmidt rank across the bipartition $(\calS,\calT)$ of a quantum state arising from the network $\Gamma$. 

In \cite{CuiShaFreSatStoMin:QuantumMaxFlowMinCut}, the authors relate the quantum max-flow to the notion of \textit{quantum min-cut}, denoted $\qmincut(\Gamma,\bfm,\calS,\calT)$, that is the information-theoretic analog of the classical min-cut in graph theory \cite{elias:classicalmincutmaxflow,ford_fulkerson_1956}, see \autoref{sec:qmaxflowproblem}. They prove the inequality $\qmaxflow(\Gamma,\mathbf{m},\mathcal{S},\mathcal{T}) \leq \qmincut(\Gamma,\mathbf{m},\mathcal{S},\mathcal{T})$ and construct examples where this inequality is tight; however, this inequality is not tight in general. They highlight connections between the quantum max-flow and entropy of entanglement as well as the quantum satisfiability problem and suggest further connections to spin systems in condensed matter and quantum gravity. Further progress was achieved in \cite{GesLanWal:MPSandQuantumMaxFlowMinCut} where the authors construct families of examples where the gap between the quantum min-cut and the quantum max-flow can be arbitrarily large. In \cite{Hastings:asymptoticsofqmincut}, Hastings showed that asymptotically, quantum min-cut and quantum max-flow are the same.

As of today, the exact value of the quantum max-flow was computed only in a limited number of cases with small fixed bond dimensions and computing the quantum max-flow exactly for a given graph seemed out of reach. In this work, we study this problem for a specific graph, the bridge graph in \autoref{fig: bridge graph}. We relate the problem of computing the quantum max-flow to the theory of prehomogeneous tensor spaces and the representation theory of quivers; this allows us to compute the max-flow for essentially all choices of bond dimensions. 
\vspace{.5cm}
\begin{figure}[!h]
  \begin{tikzpicture}[scale = .5]
 \draw[fill=black] (0,0) circle (.5cm);
 \draw[fill=black] (0,5) circle (.5cm);
 \draw[fill=black] (5,0) circle (.5cm);
 \draw[fill=black] (5,5) circle (.5cm);
 \draw[fill=black] (10,0) circle (.5cm);
 \draw[fill=black] (10,5) circle (.5cm);
 \path[draw,thick] (0,0)--(10,0);
 \path[draw,thick] (0,5)--(10,5);
 \path[draw,thick] (5,0)--(5,5);
\end{tikzpicture}
\caption{The bridge graph.} \label{fig: bridge graph}
\end{figure}
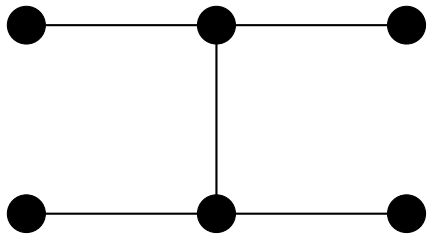

The study of the quantum max-flow in the bridge graph has two interesting connections to other areas, that we briefly outline.

In the representation theory of quivers, the value of the max-flow reveals the existence (or the non-existence) of certain covariants for the Kronecker quiver. Quiver representation theory is the main tool used in the study of the invariant theory of families of matrices \cite{DerMak:MLEMatrix}, which has, among others, applications in the study of matrix product states \cite{MicShi:QuantumVersionWielandt,DeMotSey:LinearSpanMPS}. In fact, the invariant theory of the Kronecker quiver is  the main tool used in the proof of \autoref{thm: UVW quiver}, where we rely on the existence of certain invariants introduced in \cite{Sch:SemiInvQuivers} and studied in \cite{domokos,SCHOFIELD2001125,derksenmakampolybounds}.

The second connection we briefly mention is the one to algebraic statistics and the study of maximum likelihood estimation, a widely used method to determine the free parameters of a probability distribution which best explains given data. In this context, the maximum likelihood threshold is the smallest sample size allowing one to completely reconstruct the model from the given sample, see e.g. \cite{DrtKurHof:MLEKroneckerCovariance}. The existence of a maximum likelihood threshold has strong connections with classical invariant theory, as explained in \cite{Invarianttheoryandmaximumlikelihood,DerMak:MLEMatrix}. In particular, \cite{DerMak:MLEMatrix,DerMakWal:MLETensors} study invariant theoretic properties of matrix and tensor spaces in the context of maximum likelihood estimation. Moreover, \cite{Invarianttheoryandmaximumlikelihood} studies the maximum likelihood threshold in relation to the \emph{cut-and-paste rank} introduced in \cite{BurDra:HilbNullConeMatricesBil}; in the case of matrices, the cut-and-paste rank is exactly the quantum max-flow on the bridge graph. In particular, the results of our work can be read in terms of cut-and-paste rank.

\subsection*{Acknowledgments} F.G.'s work is partially supported by the Thematic Research Programme ``Tensors: geometry, complexity and quantum entanglement'', University of Warsaw, Excellence Initiative -- Research University and the Simons Foundation Award No. 663281 granted to the Institute of Mathematics of the Polish Academy of Sciences for the years 2021--2023. We thank Mathias Drton, Alexandros Grosdos, Visu Makam and Philipp Reichenbach for helpful discussions and for pointing out the connections to algebraic statistics and to the representation theory of quivers. This work was done while V.L. was at the University of Copenhagen. V.L. and V.S. acknowledge financial support from the European Research Council (ERC Grant Agreement No. 818761), VILLUM FONDEN via the QMATH Centre of Excellence (Grant No. 10059). V.S. thanks Matthias Christandl and Frederik Ravn Klausen for helpful discussions and Or Sattath for proposing the problem.

This preprint has not undergone any post-submission improvements or corrections. The Version of Record of this article is published in \emph{Transformation Groups}, and is available online at \url{https://doi.org/10.1007/s00031-024-09863-2}

\subsection{The quantum max-flow and the bridge graph: summary of the results}\label{subsec:Tensornetworks and the quantum max-flow}

Associate to the bridge graph of \autoref{fig: bridge graph} bond dimensions $a,b,w,a',b'$ on the edges and consider the two disjoint subsets $\calS$ and $\calT$ of the vertices as follows: 
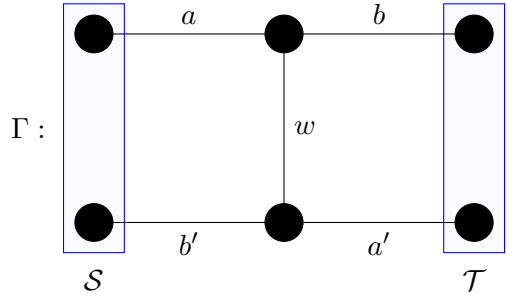
\vspace{1em}
\begin{center}
 \begin{tikzpicture}[scale = .5]
\draw[color = blue, fill = blue!2!white] ( -.8, -.8) rectangle (0.8, 5.8);
\draw[color = blue, fill = blue!2!white] ( 9.2, -.8) rectangle (10.8, 5.8);
 \draw[fill=black] (0,0) circle (.5cm);
 \draw[fill=black] (0,5) circle (.5cm);
 \draw[fill=black] (5,0) circle (.5cm);
 \draw[fill=black] (5,5) circle (.5cm);
 \draw[fill=black] (10,0) circle (.5cm);
 \draw[fill=black] (10,5) circle (.5cm);
 \path[draw] (0,0)--(10,0);
 \path[draw] (0,5)--(10,5);
 \path[draw] (5,0)--(5,5);
\node[anchor=north] () at (2.5,0) {$b'$};
\node[anchor=north] () at (7.5,0) {$a'$};
\node[anchor=south] () at (7.5,5) {$b$};
\node[anchor=south] () at (2.5,5) {$a$};
\node[anchor=west] () at (5,2.5) {$w$};
\node[anchor=north] () at (0,-1) {$\mathcal{S}$};
\node[anchor=north] () at (10,-1) {$\mathcal{T}$};
\node[anchor=east] () at (-1,2.5) {$\Gamma:$};
\end{tikzpicture}
\end{center}

Fix vector spaces $A,B,A',B',W$ of dimension $a,b,a',b',w$ respectively. A flow map on the bridge graph is defined as follows. Let $T \in A \otimes B \otimes W$ and $T' \in A' \otimes B' \otimes W^*$ be two tensors; pictorially one can think of them as placed on the top and bottom central vertex of the bridge graph, with the ``legs'' corresponding to the three edges incident to each vertex. Let $F_{T,T'} \in A \otimes B \otimes A' \otimes B'$ be the tensor obtained by contracting the factor $W$ of $T$ with the factor $W^*$ of $T'$; this is the tensor network state defined by $T$ and $T'$ on the bridge graph. The tensor $F_{T,T'}$ can be regarded as a linear map $F_{T,T'} : (A \otimes B')^*\rightarrow (B \otimes A')$, namely a bipartite state between the vertices in $\calS$ and the vertices in $\calT$. The quantum flow associated to $T$ and $T'$ is the rank of this linear map. The quantum max-flow is the maximum possible value of $\rank(F_{T,T'})$ as $T$ and $T'$ vary in the respective spaces; we write
\[
\qmaxflow \left(\bridgegraphsmall{a}{b}{w}{b'}{a'} \right) = \max \left\{ \rank(F_{T,T'}) : \begin{array}{l} T \in A \otimes B \otimes W \\ T' \in A' \otimes B' \otimes W^* \end{array}\right\}.
 \]
Note that quantum max-flow in the bridge graph is symmetric, in the sense that 
\begin{equation*}
\qmaxflow \left(\bridgegraphsmall{a}{b}{w}{b'}{a'} \right) =\qmaxflow \left(\bridgegraphsmall{b}{a}{w}{a'}{b'} \right).
\end{equation*}
Moreover, if $a \leq b$ and $b' \leq a'$, then the quantum max-flow in the bridge graph is $ab'$; this coincides with the quantum min-cut that will be introduced in \autoref{sec:qmaxflowproblem}. Hence, we will restrict our analysis to the case $a\leq b$ and $a' \leq b'$. 

A crucial role in the behaviour of the quantum max-flow is played by the constant $\lambda_w = \frac{w + \sqrt{w^2 - 4}}{2}$. In order to summarize our results, we consider a partition of the set $\{ (a,b) \in \bbN^2 : b \geq a\}$ into five regions:

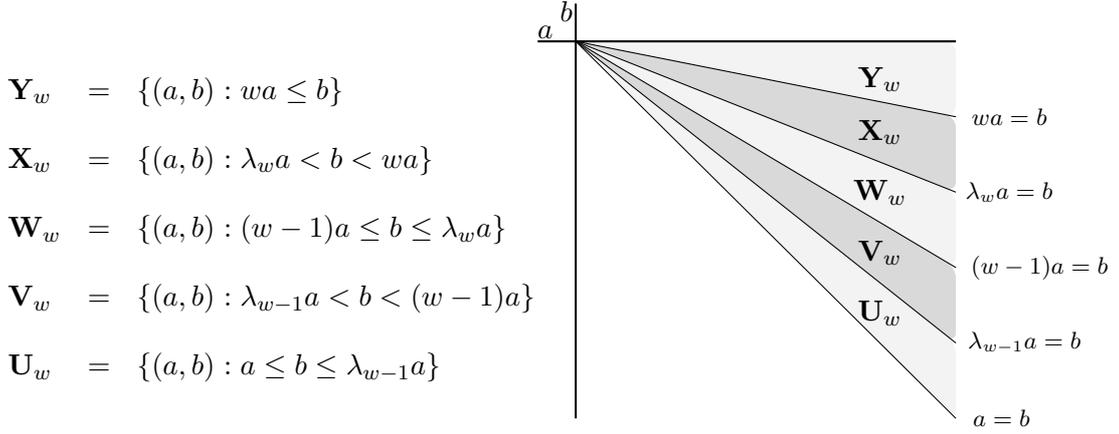
\begin{center}
    \begin{tikzpicture}
    \fill[color = {rgb:black,1;white,10}, rounded corners] (0,5) -- (5,5) -- (5,4) -- cycle;
    \fill[color = {rgb:black,1;white,10}, rounded corners] (0,5) -- (5,3) -- (5,2) -- cycle;
    \fill[color = {rgb:black,1;white,10}, rounded corners] (0,5) -- (5,1) -- (5,0) -- cycle;
    \fill[color = {rgb:black,1;white,4}, rounded corners] (0,5) -- (5,4) -- (5,3) -- cycle;
    \fill[color = {rgb:black,1;white,4}, rounded corners] (0,5) -- (5,2) -- (5,1) -- cycle;

    \draw[thick] (-0.5,5) -- (5,5);
    \draw[thick] (0,0) -- (0,5.5);
    \draw (0,5) -- (5,4);
    \draw (0,5) -- (5,3);
    \draw (0,5) -- (5,2);
    \draw (0,5) -- (5,1);
    \draw (0,5) -- (5,0);
        
    \node at (-.4,5.125) {$a$};
    \node at (-.125,5.4) {$b$};
    \node at (5.7,4) {\footnotesize $wa = b$};
    \node at (5.7,3) {\footnotesize $\lambda_wa = b$};
    \node at (6.1,2) {\footnotesize $(w-1)a = b$};
    \node at (5.9,1) {\footnotesize $\lambda_{w-1} a = b$};
    \node at (5.6,0) {\footnotesize $a = b$};
    
    \node at (4,4.5) {$\bfY_w$};
    \node at (4,3.8) {$\bfX_w$};
    \node at (4,3) {$\bfW_w$};
    \node at (4,2.2) {$\bfV_w$};
    \node at (4,1.4) {$\bfU_w$};
    \node at (-4,2.5) {$
    \begin{array}{lll}
    \bfY_w &=& \lbrace (a,b): w a \leq b \rbrace \\
    ~ \\ 
    \bfX_w &=& \lbrace (a,b): \lambda_{w}a < b < wa \rbrace \\
    ~ \\ 
    \bfW_w &=& \lbrace (a,b): (w-1)a \leq b \leq \lambda_{w}a \rbrace \\
    ~ \\ 
    \bfV_w &=& \lbrace (a,b): \lambda_{w-1}a < b < (w-1)a \rbrace \\
    ~ \\ 
    \bfU_w &=& \lbrace (a,b): a \leq b \leq \lambda_{w-1} a \rbrace \\
    \end{array}$
};
    \end{tikzpicture}
\end{center}
Note that $\bfX_w = \bfV_{w+1}$ and $\bfU_w \cup \bfV_w \cup \bfW_w = \bfU_{w+1}$. Moreover, $\bfW_2  = \bfU_3=  \{ (a, a) : a \in \bbN\}$ and $\bfU_2,\bfV_2 = \emptyset$. Our result is an almost complete characterization of the quantum max-flow in the bridge graph. The value of the max-flow is summarized in \autoref{table} and \autoref{thm:maintheoremintro}. The case marked with $\diamondsuit$ is complete, but the technical statement involves the generalized Fibonacci sequence introduced in \autoref{sec:prehomogtensorspacecastling}:
\begin{align*}
 z^{(w)}_0 &= 0, \\
 z^{(w)}_1 &= 1, \\
 z^{(w)}_{p+1} &= w \cdot z^{(w)}_p - z^{(w)}_{p-1} .
\end{align*}
In \autoref{lemma: arithmetic a b p}, we will see that if $(a,b) \in \bfX_w$, then there exist unique nonnegative integers $p,\alpha,\beta$ with $p,\beta > 0$ such that 
\[
a = z^{(w)}_p \alpha + z^{(w)}_{p+1} \beta \qquad b = z^{(w)}_{p+1} \alpha + z^{(w)}_{p+2} \beta .
\]
The integers $\alpha,\beta$ play an important  role in the characterization of the max-flow. 

The incomplete cases are the ones with white background in \autoref{table}, marked with a question mark `$?$': we prove that the conjectured value in attained in a number of cases, marked with $\bigstar$ and $\clubsuit$ in \autoref{thm:maintheoremintro}. The case $\clubsuit$ is technical and involves the definition of the \emph{castling depth}, see \autoref{def:castlingdepth}.
\begin{theorem}\label{thm:maintheoremintro}
The quantum max-flow in the bridge graph is given by the value specified in \autoref{table}. The cases marked with $\diamondsuit,\bigstar,\clubsuit$ are described below: 
\begin{itemize}
    \item [$(\diamondsuit)$] \textrm{(Complete case)} Let $(a,b), (a',b') \in \bfX_w$ and let 
		\begin{equation*}
  \begin{array}{lccl}
	 a = z^{(w)}_p \alpha + z^{(w)}_{p+1} \beta, & &  &b = z^{(w)}_{p+1} \alpha + z^{(w)}_{p+2} \beta, \\
	 a' = z^{(w)}_{p'} \alpha' + z^{(w)}_{{p'}+1} \beta', & & &b' = z^{(w)}_{{p'}+1} \alpha' + z^{(w)}_{{p'}+2} \beta',
  \end{array}
\end{equation*}
be their expression in terms of the generalized Fibonacci numbers. Then, 
\begin{equation*}
\qmaxflow \left( \bridgegraphsmall{a}{b}{w}{b'}{a'} \right)= 
\left\{ \begin{array}{ll}
	a\cdot b' & \text{if } p'>p,\\
	a\cdot b' - \beta\cdot \alpha' & \text{if } p = p', \\
	a'\cdot b & \text{if } p'<p.\\
        \end{array}\right.
\end{equation*}
\item[$(\bigstar)$] (Incomplete case) Let $(a,b), (a',b') \in \bfU_w \cup \bfV_w \cup  \bfW_w$. If $(a,b) = q \cdot (a',b')$ for some $q \in \bbQ$, then 
\[
\qmaxflow \left( \bridgegraphsmall{a}{b}{w}{b'}{a'} \right)= a \cdot b' = a' \cdot b.
\]
 \item [$(\clubsuit)$] (Incomplete case) Let  $(a,b) , (a',b')\in \bfW_w$. Then
	 \[
	\qmaxflow\left(\bridgegraphsmall{a}{b}{w}{b'}{a'}\right) = \left\{ 
	\begin{array}{ll}
a \cdot b' & \text{if }  \depth(a,b) > \depth(a',b') \\
a' \cdot b & \text{if }  \depth(a,b) < \depth(a',b'),
\end{array} \right.
\]
where $\depth$ denotes the castling depth.
\end{itemize}
\end{theorem}

\begin{table}
\renewcommand{\arraystretch}{1.8}
\setlength{\tabcolsep}{.1cm}
\centering
\begin{tabular}{ |>{\centering\arraybackslash}p{1cm}|>{\centering\arraybackslash}p{1cm} |>{\centering\arraybackslash}p{1.5cm} >{\centering\arraybackslash}p{1.5cm} >{\centering\arraybackslash}p{1.5cm}>{\centering\arraybackslash}p{1.5cm} >{\centering\arraybackslash}p{1.5cm}| }
\hline
 &&\multicolumn{5}{c|}{$(a,b)$} \\
\hline
&& $\mathbf{U}_w$&$\mathbf{V}_w$&$\mathbf{W}_w$&$\mathbf{X}_w$&$\mathbf{Y}_w$\\
\hline
&$\mathbf{U}_w$& 
    \cellcolor[RGB]{255,255,255} $\bigstar^{?} $&$\cellcolor[RGB]{60,180,75}ab'$&
    \cellcolor[RGB]{70,240,240}$ab'$&\cellcolor[RGB]{70,240,240}$ab'$&\cellcolor[RGB]{245,130,48}$ab'$ \\
&$\mathbf{V}_w$&{\cellcolor[RGB]{60,180,75}}$a'b$&\cellcolor[RGB]{255,255,255}$\bigstar ^{?}$&\cellcolor[RGB]{70,240,240}$ab'$&\cellcolor[RGB]{70,240,240}$ab'$&\cellcolor[RGB]{245,130,48}$ab'$ \\
$(a',b')$&$\mathbf{W}_w$&\cellcolor[RGB]{70,240,240}$a'b$&\cellcolor[RGB]{70,240,240}$a'b$&\cellcolor[RGB]{255,255,255}$\clubsuit + \bigstar ^{?}$ &\cellcolor[RGB]{220,190,255}$ab'$&\cellcolor[RGB]{245,130,48}$ab'$  \\
&$\mathbf{X}_w$&\cellcolor[RGB]{70,240,240}$a'b$&\cellcolor[RGB]{70,240,240}$a'b$&\cellcolor[RGB]{220,190,255}$a'b$&\cellcolor[RGB]{220,190,255}$\diamondsuit$&\cellcolor[RGB]{245,130,48}$ab'$  \\
&$\mathbf{Y}_w$&\cellcolor[RGB]{245,130,48}$a'b$&\cellcolor[RGB]{245,130,48}$a'b$&\cellcolor[RGB]{245,130,48}$a'b$&\cellcolor[RGB]{245,130,48}$a'b$&\cellcolor[RGB]{245,130,48} $aa'w$\\
\hline 
 &&\multicolumn{5}{c|}{(?): Conjectured value: $\min\left\{ a\cdot b',a'\cdot b \right\}$.} \\
\hline
\end{tabular}
\vspace{1em}
\caption{The quantum max-flow in the bridge graph with bond dimension $a,b,a',b',w$. The orange cases are solved in \autoref{thm:maxfloweasyregion}. The cyan colored cases are solved in \autoref{thm:wxvsuv}. The purple cases are solved in \autoref{thm:xvswx}, the precise formulation of $\diamondsuit$ can be found in \autoref{thm:maintheoremintro}. The green cases are solved in \autoref{cor:v versus u}. The case marked with $\clubsuit$ is partly solved in \autoref{thm:diffcastlingdegree}. The cases marked with $\bigstar$ are partly solved in \autoref{thm: UVW quiver}. A precise formulation of the results is in \autoref{thm:maintheoremintro}.}\label{table}
\end{table}

In particular, we have a full characterization of the quantum max-flow in the bridge graph for \emph{symmetric} bond dimension, namely in the case $(a,b) = (a',b')$. 
\begin{corollary}
Let $w,a,b \geq 1$. Then
\begin{itemize}
\item If $(a,b) \in \bfU_w \cup \bfV_w \cup \bfW_w$, then 
\[
\qmaxflow\left(\bridgegraphsmall{a}{b}{w}{b}{a}\right) = a b.
\]
\item If $(a,b) \in \bfX_w$, and let 
\[
	 a = z^{(w)}_p \alpha + z^{(w)}_{p+1} \beta \qquad b = z^{(w)}_{p+1} \alpha + z^{(w)}_{p+2} \beta 
\]
be its expression in terms of generalized Fibonacci numbers. Then, 
\[
\qmaxflow \left( \bridgegraphsmall{a}{b}{w}{b}{a} \right)= ab - \alpha\beta.
\]
\item If $(a,b) \in \bfY_w$, then 
\[
\qmaxflow \left( \bridgegraphsmall{a}{b}{w}{b}{a} \right)= wa^2.
\]
\end{itemize}
\end{corollary}
We point out that in all cases computed in \autoref{thm:maintheoremintro} except for the ones in the region marked by $\diamondsuit$ the value of the quantum max-flow in the bridge graph coincides with the value of the quantum min-cut, defined in \autoref{sec:qmaxflowproblem}. We prove in \autoref{prop:mincutbridge} that 
\[
\qmincut\left(\bridgegraphsmall{a}{b}{w}{b'}{a'}\right) = \min\{ aa'w, a'b,ab'\},
\]
which allows us to obtain the following result.
\begin{corollary}
 In all cases computed in \autoref{thm:maintheoremintro}, the quantum max-flow coincides with the quantum min-cut except for the case $(a,b),(a',b') \in \bfX_w$ such that the expressions in terms of generalized Fibonacci numbers given by 
		\begin{equation*}
  \begin{array}{lccl}
	 a = z^{(w)}_p \alpha + z^{(w)}_{p+1} \beta, & &  &b = z^{(w)}_{p+1} \alpha + z^{(w)}_{p+2} \beta, \\
	 a' = z^{(w)}_{p'} \alpha' + z^{(w)}_{{p'}+1} \beta', & & &b' = z^{(w)}_{{p'}+1} \alpha' + z^{(w)}_{{p'}+2} \beta',
  \end{array}
\end{equation*}
satisfy $p = p'$ and $\alpha,\alpha' > 0$.
\end{corollary}

\autoref{thm:maintheoremintro} is almost complete and we conjecture that our results can be extended to the incomplete cases:
\begin{conjecture}\label{conj:intro}
Let $(a,b), (a',b') \in \bfU_w \cup \bfV_w \cup \bfW_w$. Then 
	\begin{equation*}
	\qmaxflow\left(\bridgegraphsmall{a}{b}{w}{b'}{a'}\right) = \qmincut\left(\bridgegraphsmall{a}{b}{w}{b'}{a'}\right) = \text{min} \left\{ a\cdot b',a'\cdot b \right\}.
	\end{equation*}
\end{conjecture}
In \autoref{sec:conjebehaviour}, we prove a reduction argument allowing one to deduce \autoref{conj:intro} for any $w$ from the case $w = 3$. 

\subsection{Quantum max-flow and castling transform}

A key ingredient in the proof of the results described in \autoref{subsec:Tensornetworks and the quantum max-flow} is that the quantum max-flow is easy to control under \emph{castling transform}. The castling transform was introduced in \cite{SatKim:ClassificationIrredPrehomVS} in the study of prehomogeneous tensor spaces; in general, it defines a correspondence between orbits in tensor spaces under particular group actions. We refer to \autoref{sec:prehomogtensorspacecastling} for the precise definition. In the framework of the castling transform, we will prove the following result:
\begin{theorem}\label{thm:QMaxFlowcastlingintro}
Let $a,b,a',b',w$ be natural numbers such that $a \leq bw$ and  $a' \leq b'w.$ Then, 
\begin{align*}
	a\cdot b' - \qmaxflow \left( 
	\bridgegraphsmall{a}{b}{w}{b'}{a'} 
\right)  = 
	b\cdot (b'w - a')-\qmaxflow \left( 
	\bridgegraphbig{b}{bw-a}{w}{b'w-a'}{b'}
\right) .
\end{align*}
Moreover, if $b \leq aw$ and $b' \leq a'w$, then
\begin{align*}
	a\cdot b' - \qmaxflow \left( \bridgegraphsmall{a}{b}{w}{b'}{a'} \right) = (wa - b) \cdot a' - \qmaxflow \left( \bridgegraphbig{wa-b}{a}{w}{a'}{wa'-b'} \right) .
\end{align*}
\end{theorem}

We briefly outline how \autoref{thm:QMaxFlowcastlingintro} allows one to compute the quantum max-flow in the bridge graph. For a detailed discussion, we refer to \autoref{sec:qmaxflowinbridge}. On a high level, we use \autoref{thm:QMaxFlowcastlingintro} to reduce the computation of the quantum max-flow to one of the following two, easy cases:
\begin{enumerate}[(i)]
    \item If $b \geq a\cdot w$, that is, $(a,b) \in \mathbf{Y}_w$, then 
    \begin{equation*}
        \qmaxflow\left( \bridgegraphsmall{a}{b}{w}{b'}{a'}\right) = \min \lbrace a\cdot b', a\cdot a' \cdot w \rbrace.
    \end{equation*}
    \item If $a \leq b$ and $b' \leq a'$, then
    \begin{equation*}
        \qmaxflow\left( \bridgegraphsmall{a}{b}{w}{b'}{a'}\right) = a\cdot b'.
    \end{equation*}
\end{enumerate}

For a fixed $w$, we say that two pairs of natural numbers $(a,b)$ and $(\tilde{a},\tilde{b})$ are \textit{castling equivalent} if $(a,b)$ can be transformed into $(\tilde{a},\tilde{b})$ by a number of operations $(a,b) \mapsto (wa-b,a)$ or $(a,b) \mapsto (b,wb-a)$. 

It turns out that every pair $(a,b) \in \mathbf{X}_w$ is castling equivalent to a pair in $\mathbf{Y}_w$. Consequently, we can apply \autoref{thm:QMaxFlowcastlingintro} multiple times and use (i) to calculate the quantum max-flow. This procedure is applied in the proofs of \autoref{thm:wxvsuv} and \autoref{thm:xvswx}. 

A pair $(a,b) \in \mathbf{U}_w \cup \mathbf{V}_w \cup \mathbf{W}_w$ on the other hand is always castling equivalent to a pair $(\tilde{a},\tilde{b})$ where $\tilde{a} > \tilde{b}$. This -- in many cases -- reduces calculating the quantum max-flow in the bridge graph to the easier case (ii). This procedure is applied in the proofs of \autoref{thm:wxvsuv}, \autoref{cor:v versus u} and \autoref{thm:diffcastlingdegree}.

\subsection{Open questions}\label{subsec:Open ends}
We identify some open ends of this work.
\begin{enumerate}
	\item One open task is \autoref{conj:intro}. The reduction discussed in \autoref{sec:conjebehaviour} guarantees that the case $w =3$ is equivalent to the full conjecture.
	\item We believe that methods similar to ours can be applied to graphs that are not bridge graphs. For example, it would be interesting to see if one can calculate the quantum max-flow in translation-invariant matrix product states with periodic boundary conditions -- which is the main example in prior work on the quantum max-flow where it was used to show separations between quantum min-cut and quantum max-flow~\cite{GesLanWal:MPSandQuantumMaxFlowMinCut}.  
\end{enumerate}

\section{Preliminaries}

In this section, we provide the exact definition of the quantum max-flow for a general graph, and we prove how the general definition reduces to the one given in \autoref{subsec:Tensornetworks and the quantum max-flow} in the case of the bridge graph. Moreover, we introduce the castling framework and describe the connections to the theory of prehomogeneous vector spaces that will be useful in the proofs of the main results.

\subsection{Tensor networks and the quantum min-cut/max-flow problem}\label{sec:qmaxflowproblem}

A tensor network is defined starting from a simple graph $\Gamma =(V,E)$ together with two assignments of integer weights $\bfm: E \rightarrow \mathbb{N}$ and $\textbf{n}:V\rightarrow \mathbb{N}$ on the set of edges and vertices of $\Gamma$, respectively. The values $\textbf{m}(e)$ are called the \textit{bond dimensions} of the network and the values $\textbf{n}(v)$ are called \textit{physical dimensions}. Now, given a collection of tensors $T_v \in \left(\bigotimes_{e \in E : v \in e} \bbC^{\bfm(e)} \right) \otimes \bbC^{\bfn(v)}$, for all $v \in V$ let $T \in \bigotimes_{v \in V} \bbC^{\bfn(v)}$ be the $|V|$-party unnormalized quantum state, or equivalently tensor of order $|V|$, obtained by contracting the $T_v$'s on the tensor factors corresponding to the same edge; such a tensor $T$ is a tensor network state associated to the graph $\Gamma$ with the associated bond and physical dimensions. The set of all unnormalized quantum states that can be constructed in this way form the tensor network class associated to the graph $\Gamma$ with bond dimension $\bfm$ and physical dimension $\bfn$, denoted $\calTNS^\Gamma(\bfm,\bfn)$. 

Given a network $\Gamma$ with an assignment of bond dimensions $\bfm$, let $\mathcal{S},\mathcal{T} \subseteq V$ be two disjoint subsets; for $v \in \mathcal{S} \cup \mathcal{T}$, set $n_v = \prod_{e \ni v} m_e$ and for $v \notin \mathcal{S} \cup \mathcal{T}$, set $n_v = 1$. Now, every element $T \in \calTNS^\Gamma(\textbf{m},\bfn)$ can be regarded as a bipartite state $T \in \left(\bigotimes_{v\in \calS} \bbC^{\bfn(v)} \right) \otimes \left( \bigotimes_{v\in \calT} \bbC^{\bfn(v)}\right)$, or equivalently a linear map $F_T : \bigotimes_{v \in\mathcal{S}} {\mathbb{C}^{n_v}}^* \to \bigotimes_{v \in\mathcal{T}} \mathbb{C}^{n_v}$. The \emph{quantum max-flow} of $(\Gamma,\textbf{m})$ relative to the subsets $\mathcal{S},\mathcal{T}$ is the maximum possible Schmidt rank of this bipartite state, namely
\begin{equation*}
 \qmaxflow(\Gamma,\textbf{m}; \mathcal{S},\mathcal{T}) = \max \bigl\{ \rank (F_T) : T \in \calTNS^\Gamma(\textbf{m},\bfn) \bigr\}.
 \end{equation*}

A cut of $\Gamma$ relative to the sets $\mathcal{S},\mathcal{T}$ is a partition $V = \mathcal{A} \sqcup \mathcal{B}$ of $V$ with $\mathcal{S} \subseteq \mathcal{A}$ and $\mathcal{T}\subseteq \mathcal{B}$. The quantum capacity of a cut is 
\[
 \qcap(\mathcal{A},\mathcal{B}) = \prod_{\substack{\{v_1,v_2\} \in E \\
			      v_1 \in \mathcal{A},v_2 \in \mathcal{B}}} m_{\{v_1v_2\}}.
\]
The \emph{quantum min-cut} of $(\Gamma,\textbf{m})$ relative to the sets $\mathcal{S},\mathcal{T}$ is 
\begin{equation}\label{eq:mincut}
	\qmincut(\Gamma,\textbf{m}, (\mathcal{S},\mathcal{T})) = \min \bigl\{ \qcap(\mathcal{A},\mathcal{B}) : \mathcal{S} \subseteq \mathcal{A}, \mathcal{T}\subseteq \mathcal{B}, \mathcal{A} \cap \mathcal{B} =\emptyset \bigr\}.
\end{equation}

We state two immediate but crucial properties of the quantum min-cut and the quantum max-flow. They were first stated and shown in~\cite{CuiShaFreSatStoMin:QuantumMaxFlowMinCut}.

\begin{proposition}\label{prop:basicfactsaboutmincut}
For every network $\Gamma$ with bond dimensions $\textbf{m}$, the following holds: 
\begin{itemize}
	\item $\qmaxflow(\Gamma,\bfm;\calS,\calT)\leq\qmincut (\Gamma,\bfm;\calS,\calT)$
	\item The set of $T \in \calTNS^{\Gamma}(\mathbf{m},\bfn)$ such that $\rank(F_T) = \qmaxflow(\Gamma,\mathbf{m};\mathbf{S},\mathbf{T})$ is a dense open subset of $\calTNS^{\Gamma}(\bfm,\bfn)$, in the Zariski and in the Euclidean topology.
\end{itemize}
\end{proposition}

In general, computing the quantum max-flow is not trivial. In fact, in most cases, even proving a separation between quantum max-flow and quantum min-cut is challenging. In \cite{CuiShaFreSatStoMin:QuantumMaxFlowMinCut}, the first examples where these two quantities are not the same was presented. In \cite{GesLanWal:MPSandQuantumMaxFlowMinCut}, families of examples presenting big separations between quantum max-flow and quantum min-cut were given, but computing the quantum max-flow exactly seemed out of reach. 

In this work we want to compute $\qmincut(\Gamma,\textbf{m}, (\mathcal{S},\mathcal{T}))$ and $\qmaxflow(\Gamma,\textbf{m}, (\mathcal{S},\mathcal{T}))$ in the case of the \emph{bridge network}
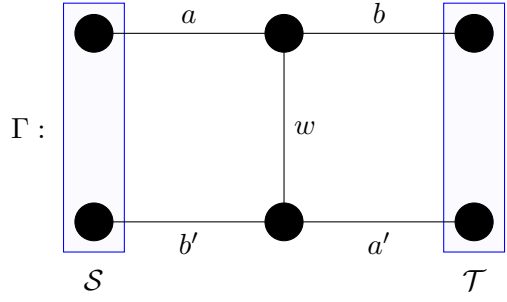
\begin{center}
 \begin{tikzpicture}[scale = .5]
\draw[color = blue, fill = blue!2!white] ( -.8, -.8) rectangle (0.8, 5.8);
\draw[color = blue, fill = blue!2!white] ( 9.2, -.8) rectangle (10.8, 5.8);
 \draw[fill=black] (0,0) circle (.5cm);
 \draw[fill=black] (0,5) circle (.5cm);
 \draw[fill=black] (5,0) circle (.5cm);
 \draw[fill=black] (5,5) circle (.5cm);
 \draw[fill=black] (10,0) circle (.5cm);
 \draw[fill=black] (10,5) circle (.5cm);
 \path[draw] (0,0)--(10,0);
 \path[draw] (0,5)--(10,5);
 \path[draw] (5,0)--(5,5);
\node[anchor=north] () at (2.5,0) {$b'$};
\node[anchor=north] () at (7.5,0) {$a'$};
\node[anchor=south] () at (7.5,5) {$b$};
\node[anchor=south] () at (2.5,5) {$a$};
\node[anchor=west] () at (5,2.5) {$w$};
\node[anchor=north] () at (0,-1) {$\mathcal{S}$};
\node[anchor=north] () at (10,-1) {$\mathcal{T}$};
\node[anchor=east] () at (-1,2.5) {$\Gamma:$};
\end{tikzpicture}
\end{center}
for integers $a,b,w,b',a'$. For the four vertices of degree one, the local tensors $T_v$ in the definition of the tensor network state have no effect. This yields the following immediate result.
\begin{lemma}
Let $\bfm = (a,b,w,b',a')$ be bond dimensions on the bridge graph. Let $\bfn$ be the physical dimensions induced by the choice of $\calS$ and $\calT$. Write $A = \bbC^a$ and similarly for $B,W,B',A'$. Then, regarded as a subset of $A \otimes B \otimes A' \otimes B'$,
\[
\calTNS^\Gamma(\bfm,\bfn) = \{ T \contract T' \in  A \otimes B \otimes A' \otimes B': T \in A\otimes B \otimes W , T'  \in A' \otimes B' \otimes W^* \} 
\]
where $\contract : W \otimes W^* \to \bbC$ is the natural tensor contraction. In particular
\[
 \qmaxflow\left(\bridgegraphsmall{a}{b}{w}{b'}{a'}\right) = \max \left\{ \rank( T \contract T'  : (A \otimes B')^* \to A' \otimes B ): \begin{array}{l} T \in A\otimes B \otimes W , \\ T'  \in A' \otimes B' \otimes W^* \end{array} \right\}.
\]
\end{lemma}

We can provide a purely linear-algebraic characterization of the quantum max-flow in the bridge graph. Fix a basis $e_1 \vvirg e_w$ of $W$ with dual basis $e^1 \vvirg e^w$. For $T \in A \otimes B \otimes W$, the \emph{slices} $T(e^j) \in A \otimes B$ can be regarded as matrices of size $a \times b$; similarly for $T' \in A' \otimes B' \otimes W^*$, the elements $T(e_j) \in A' \otimes B'$ can be regarded as matrices of size $b' \times a'$. Then, the linear map 
\[
T \contract T' : (A \otimes B')^* \to A' \otimes B
\]
is represented by the matrix $T(e^1) \boxtimes T'(e_1) + \cdots + T(e^w) \boxtimes T'(e_w) $. As $T$ and $T'$ are arbitrary we obtain the following characterization:

\begin{lemma}\label{lem:reformulationof qmaxflow}
The quantum max-flow in the bridge graph satisfies
\begin{equation*}
    \qmaxflow \left(\bridgegraphsmall{a}{b}{w}{b'}{a'} \right) = \max \left\lbrace \rank(\sum_{i = 1 }^w M_i \boxtimes N_i):  \begin{array}{c}
     M_1 \dots M_w \; \text{ matrices of size $a \times b$},     \\
     N_1 \dots N_w \; \text{ matrices of size $b' \times a'$}
    \end{array}\right\rbrace.
\end{equation*}
\end{lemma}

On the other hand, it is easy to characterize the quantum min-cut of the bridge graph:
\begin{proposition}\label{prop:mincutbridge}
Let $a,b,w,a',b'$ be integers with $a \leq b$, $a' \leq b'$. The quantum min-cut in the bridge graph is 
\[
\qmincut\left(\bridgegraphsmall{a}{b}{w}{b'}{a'} \right) = \min\{ aa'w, ab',a'b\}.
\]
\end{proposition}
\begin{proof}
The result follow from the definition, since the bridge graph has three possible cuts, of quantum capacity $aa'w, ab',a'b$ respectively. 
\end{proof}


\subsection{Prehomogenous tensor spaces and the castling transform}\label{sec:prehomogtensorspacecastling}

Let $G$ be a complex linear algebraic group acting on a complex vector space $V$. The space $V$ is \emph{prehomogeneous} for the action of $G$ if the action has a dense orbit. The theory guarantees that the orbit is Zariski open, hence dense in the Euclidean topology as well. Moreover, it is unique, and it coincides with the orbit of any \emph{generic enough} element. 

In this section, we are interested in tensor spaces $\mathbb{C}^{a}\otimes \mathbb{C}^{b}\otimes \mathbb{C}^w $ which are prehomogeneous for the action of $\GL_a \times \GL_b$ on the first two factors. We provide a sufficient condition in \autoref{prop: dense orbit}, which will be applied in \autoref{sec:qmaxflowinbridge} to compute the quantum max-flow of the bridge graph in a particular range.

A fundamental ingredient in the theory of prehomogeneous spaces is the notion of \emph{castling transform}, introduced in \cite{SatKim:ClassificationIrredPrehomVS} and extensively used in the geometric study of tensor spaces, see, e.g., \cite{Man:PrehomogeneousSpaces,Ventu:PrehomogenousTensorSpaces,DerMakWal:MLETensors}. In general, let $V$ be a representation for the group $G$. Then for every $k$, $\GL_k \times G$ naturally acts on $\bbC^k \otimes V$ and on $\bbC^k \otimes V^*$. We say that an element $T \in \bbC^k \otimes V$ is $\bbC^k$-concise if the linear map $T: (\bbC^k)^* \to V$ is injective. Set $m = \dim V$ and let $1< k < m$: the castling transform defines a natural bijection between $(\GL_k \times G)$-orbits of $\bbC^k$-concise elements in $\bbC^k \otimes V$ and $(\GL_{m-k} \times G)$-orbits of $\bbC^{m-k}$-concise elements in $\bbC^{m-k} \otimes V^*$. Notice that if $T$ is $\bbC^k$-concise in $\bbC^k \otimes V$, then $E_T = \im( (\bbC^k)^* \to V )$ is a $k$-dimensional subspace of $V$, that is an element of the Grassmannian $\Gr(k,V)$. It is easy to see that such subspace is uniquely determined by the $\GL_k$-orbit of $T$: more precisely $E_T = E_{T'}$ if and only if there exists $g \in \GL_k$ such that $g \cdot T = T'$. Consequently, $(\GL_k \times G)$-orbits of $\bbC^k$-concise elements in $\bbC^k \otimes V$ are in bijection with $G$-orbits in $\Gr(k,V)$. Similarly, $(\GL_{m-k} \times G)$-orbits of $\bbC^{m-k}$-concise elements in $\bbC^{m-k} \otimes V^*$ are in bijection with $G$-orbits in $\Gr(m-k,V^*)$. The usual identification between $\Gr(k,V)$ and $\Gr(m-k,V^*)$, given by sending a linear space $E$ to its annihilator $E^\perp$ is $\GL(V)$-equivariant, hence $G$-equivariant; in particular it preserves $G$-orbits. The castling transform is defined as the composition of these correspondences: a $\bbC^k$-concise orbit in $\bbC^k \otimes V$ is sent to the corresponding orbit in $\Gr(k,V)$, which is identified with an orbit in $\Gr(m-k,V)$ and in turn it corresponds to an orbit in $\bbC^{m-k} \otimes V^*$. Several geometric properties of the orbits are preserved under castling transform. In particular, the space $\bbC^k \otimes V$ is prehomogeneous for the action of $\GL_k \times G$ if and only if the space $\bbC^{m-k} \otimes V^*$ is prehomogeneous for the action of $\GL_k \times G$; see, e.g., \cite[Prop. 28]{Man:PrehomogeneousSpaces}.

We will use the castling correspondence in \autoref{prop: dense orbit}, but to give a precise statement we need some technical results. For every $w \geq 2$, define $\lambda_w = \frac{w + \sqrt{w^2-4}}{2}$. We have $\lambda_2 = 1$, $\lambda_w \in (w - 1, w)$ for $w \geq 3$ and $\lambda_w^{-1} = \frac{w - \sqrt{w^2-4}}{2}$; in particular $\lambda_w$ and $\lambda_w^{-1}$ are the two roots of the equation $\lambda^2 - w\lambda  + 1 = 0$. For every $w \geq 2$, define recursively the generalized Fibonacci sequence
\begin{align*}
 z^{(w)}_0 &= 0 \\
 z^{(w)}_1 &= 1 \\
 z^{(w)}_{p+1} &= w \cdot z^{(w)}_p - z^{(w)}_{p-1} .
\end{align*}
By resolving the recursion, one obtains 
\begin{align*}
 z^{(2)}_p &= p,\\
 z^{(w)}_p &= \dfrac{\lambda_w^p - \lambda_w^{-p}}{\sqrt{w^2-4}} \quad \text{if } w \neq 2 .
\end{align*}
We record an immediate fact, which will be useful multiple times throughout: 
\begin{lemma}\label{lemma: z's have det 1}
For every $w \geq 2$ and $p \geq 0$, we have 
\[
 \det \left( \begin{array}{cc}
z_{p+1} ^{(w)} & z_{p+2} ^{(w)}  \\
z_{p} ^{(w)} & z_{p+1} ^{(w)}
 \end{array}\right) = {z_{p+1} ^{(w)}} ^2 - z_{p+2} ^{(w)} \cdot z_{p} ^{(w)}= 1.
\]
\end{lemma}
\begin{proof}
 The proof is by induction on $p$. If $p = 0$, then 
\[
 \det \left( \begin{array}{cc}
z_{p+1} ^{(w)} & z_{p+2} ^{(w)}  \\
z_{p} ^{(w)} & z_{p+1} ^{(w)} \end{array}\right) = \det \left( \begin{array}{cc}
1 & w  \\
0 & 1
 \end{array}\right) = 1.
\]
For $p \geq 1$, note that
\[
\det \left( \begin{array}{cc}
z_{p+1} ^{(w)} & z_{p+2} ^{(w)}  \\
z_{p} ^{(w)} & z_{p+1} ^{(w)} \end{array}\right) = 
\det \left[\left( \begin{array}{cc}
w & -1  \\
1 & 0 \end{array}\right) \left( \begin{array}{cc}
z_{p} ^{(w)} & z_{p+1} ^{(w)}  \\
z_{p-1} ^{(w)} & z_{p} ^{(w)} \end{array}\right)\right] = \det \left( \begin{array}{cc}
z_{p} ^{(w)} & z_{p+1} ^{(w)}  \\
z_{p-1} ^{(w)} & z_{p} ^{(w)} \end{array}\right).
\qedhere
\]
\end{proof}

We also prove the following technical result, which appears in \cite{Kac:InfiniteRootSystemsRepGraphsInvariantTheory} without proof. 
\begin{lemma}\label{lemma: arithmetic a b p}
Let $w \geq 2$ and let $a,b \in \mathbb{N}$ with $\lambda_{w} a <  b$. Then there exist unique integers $\alpha,\beta \geq 0$, $\beta \neq 0$, and $p \geq 1$ such that 
 \begin{equation}\label{eq:alpha-beta-linear}
  \begin{cases}
  a &= z^{(w)}_p \alpha + z^{(w)}_{p+1} \beta\\
  b &= z^{(w)}_{p+1} \alpha + z^{(w)}_{p+2} \beta.
  \end{cases}
 \end{equation}
\end{lemma}
\begin{proof}
The equations~\eqref{eq:alpha-beta-linear} are equivalent to
\[
\left( \begin{array}{c}
\alpha  \\
\beta \end{array}\right) = 
\left( \begin{array}{cc}
z_{p+1} ^{(w)} & z_{p+2} ^{(w)}  \\
z_{p} ^{(w)} & z_{p+1} ^{(w)} \end{array}\right)^{-1}
\left( \begin{array}{c}
b  \\
a \end{array}\right) =
\left( \begin{array}{cc}
z_{p+1} ^{(w)} & -z_{p+2} ^{(w)}  \\
-z_{p} ^{(w)} & z_{p+1} ^{(w)} \end{array}\right)
\left( \begin{array}{c}
b \\
a \end{array}\right) = 
\left( \begin{array}{c}
-z_{p+2}^{(w)} a + z_{p+1}^{(w)} b \\
z_{p+1}^{(w)} a - z_{p}^{(w)} b \end{array}\right).
\]
It is clear that $\alpha, \beta$ are integer for every $p$.
The conditions $\alpha, \beta \geq 0$, $\beta \neq 0$ can now be rewritten as
\[
 \frac{z_{p}^{(w)}}{z_{p + 1}^{(w)}} < \frac{a}{b} \leq \frac{z_{p + 1}^{(w)}}{z_{p + 2}^{(w)}}.
\]
Since $\left(z_{p + 1}^{(w)}\right)^2 - z_{p+2}^{(w)} z_{p}^{(w)} = 1 > 0$, the sequence $\frac{z_{p}^{(w)}}{z_{p + 1}^{(w)}}$ strictly increases.
Therefore, there exists a unique index $p$ for which the conditions $\alpha, \beta \geq 0$, $\beta \neq 0$ hold, as long as the value $\frac{a}{b}$ lies between $\frac{z_{0}^{(w)}}{z_{1}^{(w)}} = 0$ and $\lim_{p \to \infty} \frac{z_{p}^{(w)}}{z_{p + 1}^{(w)}}$.
From the explicit formulas for $z_{p}^{(w)}$ we see that the limit is equal to $\lambda_w^{-1}$, so the required $\alpha$, $\beta$, and $p$ exist and are unique if $b > \lambda_w a$.
\end{proof}

\autoref{lemma: z's have det 1} and \autoref{lemma: arithmetic a b p} allow us to show that if $(a,b) \in \bfX_w \cup \bfY_w$, then the tensor space $\bbC^a \otimes \bbC^b\otimes \bbC^w$ is prehomogeneous for the action of $\GL_a\times \GL_b$:

\begin{proposition}\label{prop: dense orbit}
Let $w \geq 2$ and let $a,b$ be integers such that $\lambda _w a < b$. Then, $\mathbb{C}^a \otimes \mathbb{C}^b\otimes \mathbb{C}^w$ has a dense $(\GL_a \times \GL_b)$-orbit.
\end{proposition}
\begin{proof}
We prove the statement by induction on $a$.
 
If $a = 1$, we have $b > \lambda_w$, and since $\lfloor \lambda_w \rfloor = w-1$, we deduce $b \geq w$. The space $\mathbb{C}^1 \otimes \mathbb{C}^b\otimes \mathbb{C}^w$ has a dense orbit for the action of $\GL_1 \times \GL_b \simeq \GL_b$. This proves the base of the induction.

If $a \geq 2$, define $b' = a$ and $a' =  w a  - b$. Note that $\lambda_w a' < b'$. Indeed
\[
 \lambda_w a' = \lambda_w w a - \lambda_w b < \lambda_w w a -\lambda_w^2 a = a = b',
\]
where we used the condition $\lambda_w w - \lambda_w^2 = 1$.

By the induction hypothesis, the space $\mathbb{C}^{a'} \otimes \mathbb{C}^{b'}\otimes \mathbb{C}^w$ has a dense $(\GL_{a'} \otimes \GL_{b'})$-orbit. Via the castling transform, the dense orbit for the action of $\GL_{a'} \times \GL_{b'}$ on $\mathbb{C}^{a'} \otimes \mathbb{C}^{b'} \otimes \mathbb{C}^w$ corresponds to a dense orbit for the action of $\GL_{w b' - a'} \times \GL_{b'}$ on $\mathbb{C}^{w b' - a'} \otimes (\mathbb{C}^{b'} \otimes \mathbb{C}^w)^*$. Therefore the latter has a dense orbit as well. 

Since $b' = a$ and $wb' -a' = b$, we conclude that $(\bbC^{a})^* \otimes \bbC^{b} \otimes (\bbC^w)^*$ has a dense $(\GL_a \times \GL_b)$-orbit. Composing the action via any isomorphism between $(\bbC^{a})^*$ and $\bbC^a$ and between $ (\bbC^w)^*$ and $ (\bbC^w)$ concludes the proof.

\end{proof}

\subsection{Representation theory of quivers and invariants}\label{sec: quiver intro}
A quiver is a finite directed graph $Q = (Q_0,Q_1)$; $Q_0$ is the set of vertices and $Q_1$ is the set of arrows. A quiver representation of $Q$ is a pair $( (V_i)_{i \in Q_0}, (V_{e})_{e \in Q_1})$ where $V_i$ is a finite dimensional vector space and if $e = (i,j)$ is an arrow in $Q_1$ from vertex $i$ to vertex $j$, then $V_{e} : V_i \to V_j$ is a linear map. The dimension vector of the representation $V$ is $\dim V = (\dim V_i)_{i \in Q_0}$. We refer to \cite{derksen2017introduction} for a comprehensive exposition of the theory of quiver representation. In the present work, we are only interested in a series of results from \cite{Sch:SemiInvQuivers,domokos,DerkWey:SemiInvSaturationLR} applied to the case of the Kronecker quiver with $w$ arrows, that is 
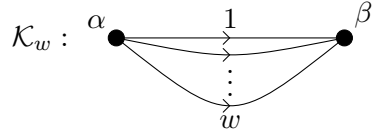
\begin{center}
\begin{tikzpicture}
\draw[->-] (0,0) -- (3,0);
\draw[->-] (0,0) .. controls (1.5,-.3) .. (3,0);
\draw[->-] (0,0) .. controls (1.5,-1.2) .. (3,0);
\draw[fill =black] (0,0) circle (3pt) node[above left] {$\alpha$};  
\draw[fill =black] (3,0) circle (3pt) node[above right] {$\beta$};  
\node[] () at (-1,0) {$\calK_w:$};
\node[anchor=south] () at (1.5,0) {$1$};
\node[anchor=north] () at (1.5,-.05) {$\vdots$};
\node[anchor=north] () at (1.5,-.9) {$w$};
\end{tikzpicture}
\end{center}
Let $a,b \geq 1$. The space of representations of $\calK_w$ with dimension vector $(a,b)$ can be identified with the space of $w$-tuples of linear maps $\bbC^a \to \bbC^b$; in other words, such space is $A \otimes B \otimes W$ with $\dim A = a$, $\dim B = b$ and $\dim W = w$; for a fixed basis $e_1 \vvirg e_w$ of $W^*$, the linear map on the arrow $j$ is $T(e_j) \in A \otimes B$. The ring of semi-invariants of $\calK_w$ with dimension vector $(a,b)$ is $\SI(\calK_w, (a,b)) \simeq \bbC[(A \otimes B \otimes W)^*]^{\SL(A) \times \SL(B)}$. A great deal of research is devoted to the study of the ring of semi-invariants for quivers. We record here a fundamental result for the Kronecker quiver, summarizing some of the results of \cite{DerkWey:SemiInvSaturationLR,domokos,SCHOFIELD2001125}. 
\begin{proposition}\label{prop: quiver basics}
Let $w \geq 3$ and let $(a,b)$ be a dimension vector for the Kronecker quiver $\calK_w$. Let $\SI(\calK_w,(a,b))^{[\delta]}$ be the component of the ring of semi-invariants $\SI(\calK_w, (a,b))$ of degree $\delta$. Then
\begin{itemize}
\item $\SI(\calK_w,(a,b))^{[\delta]}$ is spanned by the polynomials on $A \otimes B \otimes W$ of the form
\[
T \mapsto \det ( T(e_1) \boxtimes Z_1 + \cdots + T(e_w) \boxtimes Z_w) 
\]
where $T \in A \otimes B \otimes W$, and $Z_1 \vvirg Z_w$ are matrices of size $b_1 \times a_1$ such that $\delta = ab_1 = ba_1$. In particular, if $\SI(\calK_w,(a,b))_{\delta} \neq 0$, then $\delta$ is a multiple of $\lcm(a,b)$.
\item If $\SI(\calK_w,(a,b))^{[\delta]} \neq 0$ for some $\delta$, then $\SI(\calK_w,(a,b))^{[\lcm(a,b)]} \neq 0$.
\end{itemize}
\end{proposition}
\begin{proof}

The first statement is a specialization of \cite[Theorem 2.3]{SCHOFIELD2001125} to the case of Kronecker quivers. In the terminology of~\cite{SCHOFIELD2001125}, the matrices $Z_1, \dots, Z_w \in \mathbb{C}^{b_1 \times a_1}$ define an object $Z = \sum_{i = 1}^w {Z_i p_i} \colon \alpha^{\oplus b_1} \to \beta^{\oplus a_1}$ in the additive category generated by the quiver $\mathcal{K}_w$, where $p_i$ is the path from $\alpha$ to $\beta$ along the $i$-th arrow.
For a representation of $\mathcal{K}_w$ given by a tensor $T \in A \otimes B \otimes W$, the corresponding map $(A^*)^{\oplus b_1} \to B^{\oplus a_1}$ is represented by the matrix $T(e_1) \boxtimes Z_1 + \dots + T(e_w) \boxtimes Z_w$. Whenever this matrix is square of size $\delta$, its determinant is a semi-invariant of $\mathcal{K}_w$ in $\SI(\calK_w,(a,b))^{[\delta]}$. We conclude, because \cite[Theorem 2.3]{SCHOFIELD2001125} guarantees that $\SI(\calK_w,(a,b))^{[\delta]}$ is spanned by these determinants.

The proof of the second statement follows from \cite[Theorem 3]{DerkWey:SemiInvSaturationLR}; see also \cite[Theorem 10.7.8]{derksen2017introduction}. Write $\delta = ab_1 = a_1b$ for uniquely determined $a_1,b_1$; in particular, there exists $\delta_1$ such that $a_1 = \delta_1 a_2$ and $b_1 = \delta_1 b_2$ where $a_2 = \frac{a}{\gcd(a,b)}$ and $b_2 =  \frac{b}{\gcd(a,b)}$. In the notation of \cite{derksen2017introduction}, it is easy to verify that $\SI(\calK_w,(a,b))^{[\delta]} = \SI(\calK_w,(a,b))_{\sigma}$ for $\sigma = \langle (3a_1-b_1,a_1), -\rangle $. By linearity, $\sigma = \langle (3a_1-b_1,a_1), -\rangle = \delta_1 \langle (3a_2-b_2,a_2), -\rangle = \delta_1 \sigma_2$ for $\sigma_2 = \langle (3a_2-b_2,a_2), -\rangle$. \cite[Theorem 10.7.8]{derksen2017introduction} guarantees that if $\SI(\calK_w,(a,b))_{\sigma} \neq 0$ then $\SI(\calK_w,(a,b))_{\sigma_2} \neq 0$. Passing to the degrees, this guarantees $\SI(\calK_w,(a,b))^{[\delta_2]} \neq 0$, where $\delta_2 = a b_2 = a_2 b = \lcm(a,b)$.
\end{proof}


\section{The QMaxFlow in the bridge graph}\label{sec:qmaxflowinbridge}
In this section we will compute the quantum max-flow for the bridge graph for a wide range of the parameters. In \autoref{subsec:castlingandqmaxflow}, we characterize the behaviour of the quantum max-flow under castling transform. This yields the main result of this section, \autoref{thm:QMaxFlowcastling}, which will allow us to deduce the results described in \autoref{table}.

\subsection{The castling transform and the quantum max-flow}\label{subsec:castlingandqmaxflow}

In this section we will see that the quantum max-flow behaves well under castling transform. In particular, we obtain the following result, whose proof is given at the end of this section.
\begin{theorem}\label{thm:QMaxFlowcastling}
Let $a,b,a',b',w$ be natural numbers such that $a \leq bw$ and  $a' \leq b'w.$ Then, 
\begin{align*}
	a\cdot b' - \qmaxflow \left( 
	\bridgegraphsmall{a}{b}{w}{b'}{a'} 
\right)  = 
	b\cdot (b'w - a')-\qmaxflow \left( 
	\bridgegraphbig{b}{bw-a}{w}{b'w-a'}{b'}
\right) .
\end{align*}
Moreover, if $b \leq aw$ and $b' \leq a'w$, then 
\begin{align*}
	a\cdot b' - \qmaxflow \left( \bridgegraphsmall{a}{b}{w}{b'}{a'} \right) = (wa - b)\cdot  a' - \qmaxflow \left( \bridgegraphbig{wa-b}{a}{w}{a'}{wa'-b'} \right) .
\end{align*}
\end{theorem}

For a fixed $w$, we say that two pairs $(a,b)$ and $(a_0,b_0)$ are \emph{castling equivalent} if there is a finite sequence $(a_0,b_0) \vvirg (a_N,b_N)$ such that either $(a_i,b_i) = (wa_{i-1}-b_{i-1},a_{i-1})$ or $(a_i,b_i) = (b_i,w b_i-a_i)$, and $(a,b) = (a_N,b_N)$.

In the following, assume $a \leq bw$ and $a' \leq b'w$. Let $\Gr(k,V)$ denote the Grassmannian of $k$-planes in a vector space $V$. Given $(E,E') \in \Gr(a,B \otimes W) \times \Gr(a' , B' \otimes W^*)$, consider the linear map $F^{1} _{E,E'} $ defined by
\begin{align*}
 F^{1} _{E,E'} : E \otimes B'^* &\to [E' \otimes B^*]^* \\
 e \otimes \beta' &\mapsto (e' \otimes \beta \mapsto \beta(e) \contract \beta'(e')),
\end{align*}
and extended by linearity. 
\begin{proposition}\label{prop: F and F1 and F2}
 Let $ a \leq bw$ and $a' \leq b' w$ and let $T \in A \otimes B \otimes W$ and $T' \in A' \otimes B' \otimes W^*$ be such that the induced maps
\begin{equation*}
T \colon A^* \rightarrow B \otimes W \text{ and } T' \colon A'^{*} \rightarrow B' \otimes W^*
\end{equation*}
are injective. Write $E_T = \im(T: A^* \to B \otimes W)$, $E_{T'} = \im(T: A'^* \to B' \otimes W^*)$. Then,
 \[
  \rank (F_{T,T'}) = \rank (F^1_{E_T,E_{T'}})
 \]
and 
\begin{equation*}
	\dim \ker (F_{T,T'}) = \dim \ker (F^{1}_{E_T,E_{T'}}).
\end{equation*}
\end{proposition}
\begin{proof}
The map $F_{T,T'}$ factors as follows:
\begin{equation*}
  A^* \otimes B'^* \xto{T \otimes \id_{B'^*}} E_T \otimes B'^* \xto{F^1_{E_T,E_{T'}}} [E_{T'} \otimes B^*]^* \simeq E_{T'}^* \otimes B \xto{T'^{\textbf{\textit{t}}} \otimes \id_B} A' \otimes B,
\end{equation*}
 where $T'^{\textbf{\textit{t}}} : [B' \otimes W]^* \to A'$ induces naturally the map $E_{T'}^*\to A'$ because $E_{T'}^* \simeq [B' \otimes W]^* / E_{T'}^\perp$ and $E_{T'}^\perp = \ker( T'^{\textbf{\textit{t}}})$. From the injectivity assumption on the maps induced by $T$ and $T'$ we deduce that the maps $T: A^* \to E_T$ and $T'^{\textbf{\textit{t}}} : E_{T'} \to A'$ are isomorphisms. This guarantees $ \rank (F_{T,T'}) = \rank (F^1_{E_T,E_{T'}})$. The conditions on the dimension of the kernels is an immediate consequence. 
\end{proof}

The following result shows that the construction of the map $F^1_{E,E'}$ is, in some sense, equivariant under castling.
\begin{theorem}\label{thm: kernel preserving castling stuff}
For $a \leq bw$ and $a' \leq b'w$ let $(E,E') \in \Gr(a,B \otimes W) \times \Gr(a' , B' \otimes W^*)$, so that $(E'^\perp,E^\perp) \in \Gr(b'w-a' , B'^* \otimes W) \times \Gr(bw-a,B^* \otimes W^*)$. Then 
\[
 \ker ( F^1_{E,E'}) = \ker ( {F^1_{E^\perp, E'^\perp}}^{\textbf{\textit{t}}}),
\]
regarded as a subspace of $(B \otimes W) \otimes B'^* \simeq B \otimes (W^* \otimes B')^*$.
\end{theorem}
\begin{proof}
Notice $F^1_{E^\perp,E'^\perp} : E^\perp \otimes B' \to [E'^\perp \otimes B]^*$, so ${F^1_{E^\perp,E'^\perp}} ^{\textbf{\textit{t}}} : E'^\perp \otimes B \to [E^\perp \otimes B']^*$. In particular, the domain $E \otimes B'^*$ of $F^1_{E,E'}$ and the domain $E'^\perp \otimes B$ of ${F^1_{E^\perp,E'^\perp}} ^{\textbf{\textit{t}}}$ can both be regarded as subspaces of $(B \otimes W) \otimes B'^* \simeq (B' \otimes W^*)^* \otimes B$.

Let $\theta \in \ker ( F^1_{E,E'}) \subseteq (B \otimes W) \otimes B'^*$. First, we show that $\theta \in E'^\perp \otimes B$ regarded as a subspace of $(B' \otimes W^*)^* \otimes B$. Indeed, since $\theta \in \ker ( F^1_{E,E'}) $, the element $ F^1_{E,E'} (\theta) \in [E' \otimes B^*]^*$ is identically $0$ as a map $  F^1_{E,E'} (\theta): E' \otimes B^* \to \mathbb{C}$. Since $F^1_{E,E'} (\theta)$ is defined by the contraction of $\theta$ against the elements of $E' \otimes B^*$, we deduce $\theta \in (E' \otimes B^*)^\perp = E'^\perp \otimes B$. In particular ${F^1_{E^\perp,E'^\perp}}^{\textbf{\textit{t}}}$ is well defined on $\theta$.

Now observe ${F^1_{E^\perp,E'^\perp}}^{\textbf{\textit{t}}}(\theta) = 0$. Indeed, ${F^1_{E^\perp,E'^\perp}}^{\textbf{\textit{t}}}(\theta) \in [E^\perp \otimes B']^*$ is the map $E^\perp \otimes B' \to \mathbb{C}$ defined by contraction of $\theta$ against the elements of $E^\perp \otimes B'$. By assumption $\theta \in E \otimes B'^* = (E^\perp \otimes B')^\perp$, hence this contraction is identically $0$.

This shows the inclusion $\ker ( F^1_{E,E'}) \subseteq \ker ( {F^1_{E^\perp, E'^\perp}}^{\textbf{\textit{t}}})$. Reversing the roles of $(E,E')$ with $(E^\perp,E'^\perp)$, one has the other inclusion, hence equality.
\end{proof}

\autoref{thm: kernel preserving castling stuff} allows us to prove the main result of this section, namely \autoref{thm:QMaxFlowcastling}.
\begin{proof}[Proof of \autoref{thm:QMaxFlowcastling}]
It is clear that we can choose tensors $T \in A \otimes B \otimes W$ and $T' \in B' \otimes A' \otimes W^*$ with the property that they maximize the quantum flow and that the maps 
\begin{equation*}
T \colon A^* \rightarrow B \otimes W \text{ and } T' \colon A'^{*} \rightarrow B' \otimes W^*
\end{equation*}
are injective; this follows from the fact that these are open conditions in the Zariski topology. By \autoref{prop: F and F1 and F2} and \autoref{thm: kernel preserving castling stuff}, we obtain
\begin{align*}
	ab'-\qmaxflow \left( 
		\bridgegraphsmall{a}{b}{w}{b'}{a'}
	\right)  &= \text{dim}(\text{ker}(F_{T,T'})) 
				 \\&= \text{dim}(\text{ker}(F^{1}_{E_T,E_{T'}})) 
				 \\&=  \text{dim}(\text{ker}((F^{1}_{E_T^{\perp},E_{T'}^{\perp}})^\mathbf{t})) 
				 \\&=  \text{dim}(\text{ker}(F_{S,S'})) \geq b(b'w-a')-\qmaxflow \left( \bridgegraphbig{b}{bw-a}{w}{b'w-a'}{b'} \right)  
\end{align*}
where the tensors $S \in \mathbb{C}^{bw-a}\otimes \mathbb{C}^{b} \otimes \mathbb{C}^w $ and $S' \in \mathbb{C}^{b'w-a'}\otimes \mathbb{C}^{b'}\otimes \mathbb{C}^w $ are chosen so that the images of the associated linear maps $S \colon (\mathbb{C}^{bw-a})^* \rightarrow \bbC^b \otimes \bbC^w$ and $S' \colon (\mathbb{C}^{b'w-a})^* \rightarrow \bbC^{b'} \otimes \bbC^w$ are $E_T^\perp$ and $E_{T'}^\perp$, respectively. The same argument, in the opposite direction, yields equality. 

Now, assume $b \leq wa$ and $b' \leq wa'$. Set $x = wa - b, y = a , x' = wa'-b'$ and $y' = a'$ so that $x \leq wy$ and $x' \leq wy'$. Then, we obtain 
\begin{align*}
	(wa-b)a' - \qmaxflow \left( \bridgegraphbig{wa-b}{a}{w}{a'}{wa'-b'} \right)
	&=xy' - \qmaxflow \left( \bridgegraphsmall{x}{y}{w}{y'}{x'} \right) =\\
	= y(y'w-x') - \qmaxflow \left( \bridgegraphbig{y}{yw-x}{w}{y'w-x'}{y'} \right)
	&= ab' -\qmaxflow \left( \bridgegraphsmall{a}{b}{w}{b'}{a'} \right).
& \qedhere
\end{align*}
\end{proof}

\subsection{Calculating the quantum max-flow}\label{subse:calcqmaxflow} 
\autoref{thm:QMaxFlowcastling} enables us to calculate the quantum max-flow in a wide range of cases. In this section, we compute the max-flow in a series of results:~\autoref{thm:maxfloweasyregion},~\autoref{thm:wxvsuv},~\autoref{thm:xvswx}, and~\autoref{cor:v versus u}. These are obtained via arithmetic arguments: starting from pairs $(a,b),(a',b')$, one reduces via castling to ``easier pairs'' and obtains the result using~\autoref{thm:QMaxFlowcastling}.  More precisely, we will use the following observations about the behavior of tuples $(a,b)$ under castling operations which we visualize in~\autoref{fig:Regionsmaintext}. Recall the regions $\mathbf{U}_w, \mathbf{V}_w, \mathbf{W}_w, \mathbf{X}_w$ and $\mathbf{Y}_w$ from~\autoref{subsec:Tensornetworks and the quantum max-flow}.

\begin{lemma}\label{lem:behaviorundercastling}
Consider a pair $(a,b) \in \mathbb{N}^2$. As long as $a_nw \geq b_n$, define recursively
\begin{equation}\label{eq:castlingsteps}
\begin{aligned}
	(a_0,b_0) &= (a,b), \\
 (a_{n + 1}, b_{n +1} ) &= (wa_n - b_n, a_n).
 \end{aligned}
\end{equation}
\begin{enumerate}[(a)]
	\item\label{castlinstepsitema}   Let $(a,b) \in \mathbf{Y}_w$. Then, $(b,bw-a) \in \bfX_w$, that is, with one castling step, we move to the region $\mathbf{X}_w$.
	\item\label{castlinstepsitemb} Let $(a,b) \in \mathbf{X}_w$ and write 
		\begin{equation*}
		  a = z^{(w)}_p \alpha + z^{(w)}_{p+1} \beta, \;
  b = z^{(w)}_{p+1} \alpha + z^{(w)}_{p+2} \beta. 
		\end{equation*}
		as in~\autoref{lemma: arithmetic a b p}. Then, for all $n = 1, \dots , p-1$, we have $(a_n,b_n) \in \mathbf{X}_w$ as well as $(a_p,b_p) \in \mathbf{Y}_w$. That is, with $p$ castling steps, we reach $\mathbf{Y}_w$ and stay until then in $\mathbf{X}_w$.
	\item\label{castlinstepsitemc} Let $(a,b) \in \mathbf{W}_w$ and consider the recursive sequence in~\eqref{eq:castlingsteps}. Then, for some $n \in \mathbb{N}$, we have $(a_n,b_n) \in \mathbf{U}_w \cup \mathbf{V}_w$. Moreover, we have $(b,bw-a)\in \mathbf{W}_w$, in other words, castling ``in the other direction'' lets us stay in $\mathbf{W}_w$.
	\item\label{castlinstepsitemd} Let $(a,b) \in \mathbf{U}_w \cup \mathbf{V}_w$. Then, $(b,wb-a) \in \mathbf{W}_w$ and $(a,wa-b) \in \mathbf{U}_w \cup \mathbf{V}_w$. In particular, for $w \geq 4$ and $(a,b) \in \bfV_w$, it holds $(a,wa-b) \in \mathbf{U}_w$ and similarly, for $(a,b) \in \bfU_w$, it holds $(a,wa-b) \in \bfV_w$.
\end{enumerate}
\end{lemma}
For any fixed $w$, we say that two pairs $(a,b)$ and $(a',b')$ are \emph{castling equivalent} if there is a sequence as in~\autoref{eq:castlingsteps} with $(a_0,b_0) = (a,b)$ and $(a_n,b_n) = (a',b')$. 
\begin{proof}
Note that in all cases, $a \neq 0$. Start with~\autoref{castlinstepsitema} and let $(a,b) \in \mathbf{Y}_w$. Then, $bw-a < bw$ as well as $(w-\lambda_w)b \geq (w^2 - w\lambda_w)a \geq a$ which shows $(bw-a,b) \in \mathbf{X}_w$. For~\autoref{castlinstepsitemb}, we observe that if $(a,b)\in \mathbf{X}_w$ can be written as
\begin{equation*}
		  a = z^{(w)}_p \alpha + z^{(w)}_{p+1} \beta, \;
  b = z^{(w)}_{p+1} \alpha + z^{(w)}_{p+2} \beta
\end{equation*}
as in~\autoref{lemma: arithmetic a b p}, then $wa-b = z^{(w)}_{p-1} \alpha + z^{(w)}_{p} \beta$ and $a = z^{(w)}_{p} \alpha + z^{(w)}_{p+1} \beta$.
Since $z_{i}^{(w)} \neq 0$ as long as $i \neq 0$, the claim follows by applying this observation $p$ times. For~\autoref{castlinstepsitemc}, assume $(a,b) \in \mathbf{W}_w$.  We have  $a = (\lambda_w w - \lambda_w^2 )a  = \lambda_w (wa - \lambda_w a) < \lambda_w (wa -b)$ and $wa - b \leq wa - (w-1)a = a$, in other words, $(wa-b,a) \in \mathbf{U}_w \cup \mathbf{V}_w \cup \mathbf{W}_w$. Moreover, we see that, by assumption, the second coordinate strictly decreases in the castling step, that is, $a <b$. This guarantees that the sequence cannot stay in $\mathbf{W}_w$ forever and, in particular, that there is an $n$ such that $(a_n,b_n) \in \mathbf{U}_w \cup \mathbf{V}_w$. Finally, for $(a,b)\in \mathbf{U}_w\cup\mathbf{V}_w$,  we clearly have $(w-1)b = wb-b \leq wb-a$. From~\autoref{castlinstepsitema} and~\autoref{castlinstepsitemb}, we know that $(b,wb-a)\not\in\mathbf{X}_w \cup \mathbf{Y}_w$ and therefore, $(b,wb-a)\in \mathbf{W}_w$. Moreover, if $w \geq 4$ and  $(a,b) \in \mathbf{V}_w$, we have $a = wa - (w-1)a \leq wa - b$ as well as $wa - b \leq wa - \lambda_{w-1} a = (w - \lambda_{w-1})a \leq \lambda_{w-1} a$ where the last inequality follows from the fact $\lambda_{w-1} \geq w-2$. The case $(a,b) \in \mathbf{U}_w$ is similar. This finishes~\autoref{castlinstepsitemd}.
\end{proof}

\definecolor{Xpurple}{rgb}{0.54, 0.17, 0.89}
\definecolor{Wgreen}{rgb}{0.13, 0.55, 0.13}
\begin{figure}
	     \centering
        \begin{tikzpicture}[scale = 1.3]
    \fill[color = {rgb:black,1;white,10}, rounded corners] (0,5) -- (5,5) -- (5,4) -- cycle;
    \fill[color = {rgb:black,1;white,10}, rounded corners] (0,5) -- (5,3) -- (5,2) -- cycle;
    \fill[color = {rgb:black,1;white,10}, rounded corners] (0,5) -- (5,1) -- (5,0) -- cycle;
    \fill[color = {rgb:black,1;white,4}, rounded corners] (0,5) -- (5,4) -- (5,3) -- cycle;
    \fill[color = {rgb:black,1;white,4}, rounded corners] (0,5) -- (5,2) -- (5,1) -- cycle;
    \fill[color = {rgb:black,1;white,10}, rounded corners] (0,5) -- (5,0) -- (4,0) -- cycle;
    \fill[color = {rgb:black,1;white,10}, rounded corners] (0,5) -- (3,0) -- (2,0) -- cycle;
    \fill[color = {rgb:black,1;white,10}, rounded corners] (0,5) -- (1,0) -- (0,0) -- cycle;
    \fill[color = {rgb:black,1;white,4}, rounded corners] (0,5) -- (4,0) -- (3,0) -- cycle;
    \fill[color = {rgb:black,1;white,4}, rounded corners] (0,5) -- (2,0) -- (1,0) -- cycle;
    \fill[color = white] (4.95,5) -- (5.1,5) -- (5.1,0) -- (4.95,0) -- cycle;
    \draw[thick] (-0.5,5) -- (5,5);
    \draw[thick] (0,0) -- (0,5.5);
    \draw (0,5) -- (5,4);
    \draw (0,5) -- (5,3);
    \draw (0,5) -- (5,2);
    \draw (0,5) -- (5,1);
    \draw[thick] (0,5) -- (5,0);
    \draw (0,5) -- (1,0);
    \draw (0,5) -- (2,0);
    \draw (0,5) -- (3,0);
    \draw (0,5) -- (4,0);
    \def\x{1};
    \node[scale = \x] at (-.4,5.18) {$a$};
    \node[scale = \x] at (-.17,5.4) {$b$};
    \node[scale = \x] at (5.75,4) {$wa = b$};
    \node[scale = \x] at (5.7,3) {$\lambda_wa = b$};
    \node[scale = \x] at (6.1,2) {$(w-1)a = b$};
    \node[scale = \x] at (5.9,1) {$\lambda_{w-1} a = b$};
    \node[scale = \x] at (5.6,0) {$a = b$};
    \node[scale = \x] at (4,4.5) {$\mathbf{Y}_w$};
    \node[scale = \x] at (4,3.8) {$\mathbf{X}_w$};
        \node[scale = \x] at (4,2.2) {$\mathbf{V}_w$};
    \node[scale = \x] at (4,1.4) {$\mathbf{U}_w$};
       \node[scale = \x] at (3.6,1) {$\mathbf{U}'_w$};
    \node[scale = \x] at (2.8,1) {$\mathbf{V}'_w$};
    
    \node[scale = \x] at (1.2,1) {$\mathbf{X}'_w$};
    \node[scale = \x] at (.5,1) {$\mathbf{Y}'_w$};
    \draw[fill = Xpurple,color = Xpurple] (1.1,4.9) circle (.06);
    \draw[fill = Xpurple,color = Xpurple] (1.8,4.45) circle (.06);
    \draw[fill = Xpurple,color = Xpurple] (3,3.9) circle (.06);
    \draw[fill = Xpurple,color = Xpurple] (4.5,3.3) circle (.06);
    \draw[color = Xpurple] (1.1,4.9) -- (1.8,4.45) -- (3,3.9) -- (4.8,3.18);
    \draw[color = Wgreen,fill = Wgreen] (4.6,3) circle (.06);
    \draw[color = Wgreen,fill = Wgreen] (3,3) circle (.06);
    \draw[color = Wgreen,fill = Wgreen] (2,2.8) circle (.06);
    \draw[color = Wgreen,fill = Wgreen] (1.7,1.8) circle (.06);
    \draw[color = Wgreen,fill = Wgreen] (2.1,0.2) circle (.06);
    \draw[color = Wgreen] (4.8,2.97) -- (4.6,3) --(3,3) -- (2,2.8) -- (1.7,1.8) -- (2.1,0.2) -- (2.15,0.1);

\fill[color = {rgb:black,1;white,10}, rounded corners] (2,1) circle (0.3);
\fill[color = {rgb:black,1;white,10}, rounded corners] (4,3) circle (0.3);
    \node[scale = \x] at (2,1) {$\mathbf{W}'_w$};
    \node[scale = \x] at (4,3) {$\mathbf{W}_w$};
    \end{tikzpicture}
    \caption{The regions $\mathbf{U}_w, \mathbf{V}_w, \mathbf{W}_w, \mathbf{X}_w$ and $\mathbf{Y}_w$ with their counterparts ``on the other side'' of the line $a = b$. In purple, we visualize the cases~\eqref{castlinstepsitema} and~\eqref{castlinstepsitemb} in~\autoref{lem:behaviorundercastling}: Any pair in $\mathbf{Y}_w$ castles directly to the $\mathbf{X}_w$ region. For any pair in the $\mathbf{X}_w$ region, we eventually land in the $\mathbf{Y}_w$ region by castling. Cases~\eqref{castlinstepsitemc} and~\eqref{castlinstepsitemd} are visualized in green. Here, we see that any pair in $\mathbf{W}_w$ castles eventually to a pair in  $\mathbf{U}_w \cup \mathbf{V}_w$ and then ``flips side'' to $\mathbf{U}_w' \cup \mathbf{V}_w'$.}\label{fig:Regionsmaintext} 
    \end{figure}

\begin{definition}\label{def:castlingdepth}
	Let $(a,b) \in \mathbf{U}_w \cup \mathbf{V}_w \cup \mathbf{W}_w$ and consider the sequence $(a_n,b_n)$ in~\autoref{eq:castlingsteps}. Call the minimal $n_0$ such that $(a_{n_0},b_{n_0}) \in \mathbf{U}_w \cup \mathbf{V}_w$ the \emph{castling depth} of $(a,b)$, denoted $\mathrm{depth}(a,b)$. 
\end{definition}

The existence of such a minimal $n_0$ is guaranteed by~\autoref{lem:behaviorundercastling}. In particular, for all $(a,b) \in \mathbf{U}_w \cup \mathbf{V}_w$, we have $\mathrm{depth}(a,b) = 0$.

We will now calculate the quantum max-flow in the bridge graph. We start with the easiest case when $(a,b) \in \mathbf{Y}_w$, that is, $ b \geq aw$. This corresponds to the orange cells in~\autoref{table}.

\begin{theorem}\label{thm:maxfloweasyregion}
Let $a,b,a',b',w$ be natural numbers such that $aw \leq b$. We have 
\begin{align*}
	\qmaxflow \left( \bridgegraphsmall{a}{b}{w}{b'}{a'} \right)= \qmincut \left( \bridgegraphsmall{a}{b}{w}{b'}{a'} \right)= \left\{\begin{array}{ll}
		a b' & \text{if } (a',b') \notin \mathbf{Y}_w ,\\
		aa' w  &  \text{if }(a',b') \in \mathbf{Y}_w.
	\end{array} \right.
\end{align*}
In other words, if $T$ and $T'$ are tensors realizing the quantum max-flow, the resulting linear map $F_{T,T'}$ has zero kernel if $b' \leq a'w$ and a kernel of dimension $a\cdot (b'-a'w)$ if $b' \geq a'w$.
\end{theorem}
\begin{proof}
Let $T$ be a tensor such that the map $T: (A \otimes W)^* \rightarrow B$ is injective. If $(a',b') \notin \mathbf{Y}_w$, let $T'$  be a tensor such that the map $T' : {B'} ^* \rightarrow W^* \otimes A'$ is injective. If $(a',b') \in \mathbf{Y}_w$, let $T'$ be a tensor such that the map $T' : {B'}^* \rightarrow W^* \otimes A'$ is surjective. In both cases, we obtain the desired result.
\end{proof}

Next we consider the case $(a,b) \in \mathbf{W}_w \cup \mathbf{X}_w$ and $(a',b') \in \mathbf{U}_w \cup \mathbf{V}_w$, corresponding to the cyan cells in~\autoref{table}.
\begin{theorem}\label{thm:wxvsuv}
Let  $(a,b) \in \mathbf{W}_w \cup \mathbf{X}_w$ and $(a',b') \in \mathbf{U}_w \cup \mathbf{V}_w$, that is, $(w-1)a \leq b \leq wa$ and $a' \leq b' \leq (w-1)a'$. Then, we have
\begin{equation*}
\qmaxflow \left( \bridgegraphsmall{a}{b}{w}{b'}{a'} \right)= \qmincut \left( \bridgegraphsmall{a}{b}{w}{b'}{a'} \right)= ab'.
\end{equation*}
\end{theorem}
\begin{proof}
Since $ab' \leq aa'(w-1) \leq ba'$, it is clear the quantum min-cut, in this case, is $ab'$. Moreover since $b \leq aw$ and $b' \leq a'w$,  we can apply~\autoref{thm:QMaxFlowcastling} yielding
\begin{align}\label{eq:proofwxvsuv}
	ab' - \qmaxflow \left( \bridgegraphsmall{a}{b}{w}{b'}{a'} \right) = (wa - b) a' - \qmaxflow \left( \bridgegraphbig{wa-b}{a}{w}{a'}{wa'-b'} \right) .
\end{align}
We notice that $wa-b\leq wa - (w-1)a = a$ and $wa'-b' \geq wa' - (w-1)a' = a'$ (see also~\autoref{fig:Regionsmaintext}). This guarantees that the quantum max-flow with bond dimensions $(wa-b,a)$ and $(wa'-b',a')$ equals $(wa - b) a'$. In particular, the right-hand side of~\autoref{eq:proofwxvsuv} is $0$, showing that the left-hand side is $0$, as well. This concludes the proof.
\end{proof}

Let us turn our attention to the case $(a,b) \in \mathbf{X}_w$. This is the most delicate case and the only one where we can prove the existence of sets of bond dimensions for which the quantum max-flow is strictly smaller than the quantum min-cut. From~\autoref{thm:wxvsuv}, we already know the quantum max-flow for $(a',b') \in \mathbf{U}_w \cup \mathbf{V}_w$. From~\autoref{thm:maxfloweasyregion}, we know the quantum max-flow when $(a',b') \in \mathbf{Y}_w$. The following result provides an answer when $(a',b')\in \mathbf{W}_w \cup \mathbf{X}_w$. These are the purple cases of~\autoref{table}.

\begin{theorem}\label{thm:xvswx}
Let $(a,b) \in \mathbf{X}_w$, that is, $\lambda_w a < b \leq wa$, and write
\begin{equation*}
a = z^{(w)}_p \alpha + z^{(w)}_{p+1} \beta, \;\;b = z^{(w)}_{p+1} \alpha + z^{(w)}_{p+2} \beta 
\end{equation*}
as in~\autoref{lemma: arithmetic a b p}.  
\begin{enumerate}
	\item If $(a',b') \in \mathbf{W}_w$, then 
		\begin{equation*}
\qmaxflow \left( \bridgegraphsmall{a}{b}{w}{b'}{a'} \right)= \qmincut \left( \bridgegraphsmall{a}{b}{w}{b'}{a'} \right)= ab'.
\end{equation*}
	\item If $(a',b') \in \mathbf{X}_w$, write 
\begin{equation*}
	 a' = z^{(w)}_{p'} \alpha' + z^{(w)}_{{p'}+1} \beta', \;\;b' = z^{(w)}_{{p'}+1} \alpha' + z^{(w)}_{{p'}+2} \beta'.
\end{equation*}
 Then,
\begin{equation*}
\qmaxflow \left( \bridgegraphsmall{a}{b}{w}{b'}{a'} \right)= 
\left\{ \begin{array}{ll}
	a b' & \text{if } p'>p,\\
	a b' - \beta \alpha' & \text{if } p = p', \\
	a' b & \text{if } p'<p.\\
        \end{array}\right.
\end{equation*}
\end{enumerate}
\end{theorem}
\begin{proof}
Consider the sequence $(a_n,b_n)$ in~\autoref{eq:castlingsteps} consisting of castling equivalent pairs and define a similar sequence $(a'_n,b'_n)$. We know that $a_p = \beta$ and $b_p = \alpha + w\beta$. If $(a',b') \in \mathbf{W}_w$, we certainly have $b_p' \leq w a_p'$ and consequently
\begin{equation*}\label{eq:proofxw}
	ab' - \qmaxflow \left( \bridgegraphsmall{a}{b}{w}{b'}{a'} \right) = a_p b_p' - \qmaxflow \left( \bridgegraphsmallsub{a_p}{b_p}{w}{b_p'}{a_p'} \right) = 0
\end{equation*}
by~\autoref{thm:maxfloweasyregion}. This shows the first claim.

Now, let $(a',b')\in \mathbf{X}_w$ and assume first that $p' > p$. By~\autoref{lem:behaviorundercastling}, we see that in this case, $wa'_p \geq b_p'$ holds. The same argument as for $(a',b')\in \mathbf{W}_w$ finishes this case. We note that the case $p' <p$ follows by symmetry. 

Finally, assume $(a',b')\in \mathbf{X}_w$ and $p = p'$. In this case, applying~\autoref{thm:QMaxFlowcastling} $p$ times yields 
\[
	ab' - \qmaxflow \left( \bridgegraphsmall{a}{b}{w}{b'}{a'} \right) = \beta(\alpha'+w\beta') - \qmaxflow \left( \bridgegraphbig{\beta}{\alpha+w\beta}{w}{\alpha'+w\beta'}{\beta'} \right)  = \beta \alpha'
\]
where the last equality is~\autoref{thm:maxfloweasyregion}. This finishes the proof. 
\end{proof}

We conclude this section with an immediate consequence of the previous results, which allows us to compute the quantum max-flow when $(a,b) \in \mathbf{V}_w$ and $(a',b') \in \mathbf{U}_w$. These are the green cases of~\autoref{table}.
\begin{corollary}\label{cor:v versus u}
Let $(a,b) \in \mathbf{V}_w$ and $(a',b') \in \mathbf{U}_w$. Then,  
\begin{equation*}
\qmaxflow \left( \bridgegraphsmall{a}{b}{w}{b'}{a'} \right)= \qmincut \left( \bridgegraphsmall{a}{b}{w}{b'}{a'} \right)= ab'.
\end{equation*}
\end{corollary}
\begin{proof}
Let $T \in \bbC^a \otimes \bbC^b \otimes \bbC^{w-1}$ be generic and regard it as a tensor in $\bbC^a \otimes \bbC^b \otimes \bbC^{w}$. Similarly, let $T' \in \bbC^{a'} \otimes \bbC^{b'} \otimes \bbC^{w-1}$ be a generic element, regarded as a tensor in $\bbC^{a'} \otimes \bbC^{b'} \otimes \bbC^{w}$. The quantum max-flow is bounded from below by the rank of the flow map $F_{T,T'}$. 

By genericity of $T,T'$, the rank of $F_{T,T'}$ is the max-flow of the bridge graph in the case $w-1$. Notice $\mathbf{V}_w = \mathbf{X}_{w-1}$ and $\mathbf{U}_w = \mathbf{U}_{w-1}\cup \mathbf{V}_{w-1} \cup \mathbf{W}_{w-1}$, therefore in this case the result follows from \autoref{thm:wxvsuv} and~\autoref{thm:xvswx} applied to the case $w-1$. 
\end{proof}

\subsection{Partial solution of the remaining cases}\label{sec:partial cases}

In this section, we see a series of partial results that cover a great number of cases in which $(a,b)$ and $(a',b')$ are both in $\mathbf{W}_w$, in $\mathbf{V}_w$ or in $\mathbf{U}_w$. 

Our first observation is that for pairs with different castling depths as defined in~\autoref{def:castlingdepth}, the quantum max-flow and the quantum min-cut in the bridge graph coincide. 

\begin{theorem}\label{thm:diffcastlingdegree}
Let $(a,b),(a',b')\in \mathbf{W}_w$ be such that $\mathrm{depth}(a,b) > \mathrm{depth}(a',b')$. Then, 
\begin{equation*}
\qmaxflow \left( \bridgegraphsmall{a}{b}{w}{b'}{a'}\right) = \qmincut \left( \bridgegraphsmall{a}{b}{w}{b'}{a'}\right) =ab'.
\end{equation*}
\end{theorem}
\begin{proof}
Let $p = \mathrm{depth}(a',b')$. We apply iteratively~\autoref{thm:QMaxFlowcastling} to obtain 
\begin{equation*}
ab' - \qmaxflow \left( \bridgegraphsmall{a}{b}{w}{b'}{a'}\right)
= a_{p} b_{p}' -  \qmaxflow \left( \bridgegraphsmallsub{a_{p}}{b_{p}}{w}{b_{p}'}{a_{p}'}\right).
\end{equation*}
By assumption, $(a_{p} , b_{p}) \in \mathbf{W}_w$ and $(a_p',b_p') \in \mathbf{U}_w \cup \mathbf{V}_w$. By~\autoref{thm:wxvsuv}, we deduce that the right-hand side of the equation above is $0$. So is the left-hand side, which provides the assertion.
\end{proof}

We can also calculate the quantum max-flow for some instances of $(a,b)$ and $(a',b')$ where both have the same castling depth by generalizing~\autoref{cor:v versus u}. 
\begin{lemma}\label{lem:easywhalflemma}
Fix $a,b,w,a',b'$ such that $(a,b),(a',b') \in \mathbf{U}_w \cup \mathbf{V}_w$ and let $\overline{w} \in \lbrace 2, \dots , w-1 \rbrace$. Assume that $a\overline{w} \leq b$ and $a'\overline{w} \geq b'$. Then, 
\begin{equation}\label{eq:omegahalf} 
	\qmaxflow \left( \bridgegraphsmall{a}{b}{w}{b'}{a'}\right) = \qmincut \left( \bridgegraphsmall{a}{b}{w}{b'}{a'}\right) = ab'.
\end{equation}
\end{lemma}
\begin{proof}
It follows from~\autoref{thm:maxfloweasyregion} that~\autoref{eq:omegahalf} holds even when we replace the bridge width $w$ by $\overline{w}$. This implies the claim.   
\end{proof}

For example, recall the recursively defined sequence $(a_n,b_n)$ in~\autoref{eq:castlingsteps} and define $(a',b') = (a, wa-b)$. If $(a,b)\in \mathbf{U}_w \cup \mathbf{V}_w$, then so is $(a', b')$.
According to~\autoref{lem:behaviorundercastling}, after one castling step the reversed tuples $(b_1, a_1), (b'_1, a'_1)$ again lie $\mathbf{U}_w \cup \mathbf{V}_w$, and the subsequent reversed tuples stay in $\mathbf{W}_w$.
Reversing order of the tuple elements also reverses the castling steps, so $\mathrm{depth}(b_{n+1}, a_{n+1}) = \mathrm{depth}(b'_{n+1}, a'_{n+1}) = n$ for all $n \in \mathbb{N}$.
It turns out that we can calculate the quantum max-flow for these pairs. 

\begin{corollary}\label{thm:samecastlingdegree}
	Let $w$ be even, $(a,b) \in \mathbf{U}_w \cup \mathbf{V}_w$ and consider $(a_n,b_n)$ and $(a'_{n},b'_{n})$ as before. Then, 
\begin{equation*}
	\qmaxflow \left( \bridgegraphsmall{a_n}{b_n}{w}{b'_{n}}{a'_{n}}\right) = \qmincut \left( \bridgegraphsmall{a_n}{b_n}{w}{b'_{n}}{a'_{n}}\right) =\min  (a_n b'_n, a'_n b_n ).
\end{equation*}
\end{corollary}
\begin{proof}
Using~\autoref{thm:QMaxFlowcastling} it is clear that it suffices to show the claim for $n = 0$. There, both pairs are in $\mathbf{U}_w \cup \mathbf{V}_w$. Assuming $\frac{w}{2}a \leq b$, that is, $aw-b \leq \frac{w}{2}a$, the claim is exactly~\autoref{lem:easywhalflemma}. The case $\frac{w}{2}a \geq b$ follows by symmetry. 
\end{proof}

We conclude this section with a result that determines the quantum max-flow for proportional pairs of bond dimensions in the region $\mathbf{U}_w \cup \mathbf{V}_w \cup \mathbf{W}_w$. 
\begin{theorem}\label{thm: UVW quiver}
Let $w \geq 3$ and let $(a,b),(a',b')\in \mathbf{U}_w \cup \mathbf{V}_w \cup \mathbf{W}_w$ and suppose there exists $q \in \mathbb{Q}$ with $(a,b) = q(a',b')$. Then,
\[
\qmaxflow \left( \bridgegraphsmall{a}{b}{w}{b'}{a'}\right) = \qmincut \left( \bridgegraphsmall{a}{b}{w}{b'}{a'}\right) =ab'.
\]
\end{theorem}
\begin{proof}
The hypothesis guarantees $\dim (A \otimes B') = \dim (B \otimes A')$. In particular, the statement is equivalent to the fact that for some $T \in A \otimes B \otimes W$ and $T' \in A' \otimes B' \otimes W^*$, the flow map $F_{T,T'}$ has nonzero determinant. 

If $(a,b) \in \mathbf{U}_w \cup \mathbf{V}_w \cup \mathbf{W}_w$, a generic tensor $T \in A \otimes B \otimes W$ is \emph{semistable} for the action of $\mathrm{GL}(A) \times \mathrm{GL}(B)$ in the sense of geometric invariant theory, see, for example,~\cite[Corollary~4.10]{Invarianttheoryandmaximumlikelihood},~\cite[Proposition~1.3]{DrtKurHof:MLEKroneckerCovariance}, or~\cite[Theorem~1.2]{DerMak:MLEMatrix}. In particular, semistability of generic elements guarantees that the ring of invariants 
\[
\mathbb{C}[A \otimes B \otimes W]^{\mathrm{SL}(A) \times \mathrm{SL}(B)}
\]
is nontrivial. By the discussion of~\autoref{sec: quiver intro}, this ring of invariants coincides with $\mathrm{SI}(\mathcal{K}_w,(a,b))$. By~\autoref{prop: quiver basics}, this ring has nontrivial elements in degree $ab'$ and they arise as determinants 
\[
\det(T(\mathrm{e}_1) \boxtimes Z_1 + \cdots + T(\mathrm{e}_w) \boxtimes Z_w)
\]
for certain, generic enough, matrices $Z_j$ of size $b' \times a'$. Regarding the $w$-tuple $(Z_1  , \dots ,  Z_w)$ as a tensor in $A' \otimes B' \otimes W^*$, we see that $T(\mathrm{e}_1) \boxtimes Z_1 + \cdots + T(\mathrm{e}_w) \boxtimes Z_w$ can be identified with the flow map $F_{T,T'}$. Since this determinant is nonzero for generic enough $T$ and $Z_1  , \dots ,  Z_w$, we conclude that $F_{T,T'}$ has full rank. This yields the desired result.
\end{proof}

\subsection{Conjectural behaviour of the quantum max-flow in the bridge graph and a reduction to $w = 3$}\label{sec:conjebehaviour} 

We expect the quantum max-flow to be always equal to quantum min-cut with the only exception discussed in case $\diamondsuit$ of \autoref{thm:maintheoremintro}, which was analyzed in \autoref{thm:xvswx}. This has been proved for a large number of cases, as shown in \autoref{table}. In the cases remained open, both $(a,b)$ and $(a',b')$ belong to $\bfU_w$, $\bfV_w$ or $\bfW_w$ and the results of \autoref{sec:partial cases} do not apply. Therefore, we propose the following conjecture:
\begin{conjecture}\label{conj:generalconj}
Let $(a,b),(a',b') \in \bfU_w \cup \bfV_w \cup \bfW_w$. Then
\begin{equation*}
\qmaxflow \left( \bridgegraphsmall{a}{b}{w}{b'}{a'}\right) = \qmincut \left( \bridgegraphsmall{a}{b}{w}{b'}{a'}\right) =\min\{a'b,ab'\}.
\end{equation*}
\end{conjecture}

In this section, we prove a reduction result showing that \autoref{conj:generalconj} is equivalent to its specialization to the case $w=3$:
\begin{conjecture}\label{conj: case w = 3}
Let $(a,b), (a',b') \in \bfU_3 \cup \bfV_3 \cup \bfW_3$. Then
\[
	\qmincut \left( \bridgegraphsmall{a}{b}{3}{b'}{a'} \right) = \qmaxflow \left( \bridgegraphsmall{a}{b}{3}{b'}{a'} \right) = \min \lbrace a' b, ab'\rbrace.
\]
\end{conjecture}

It is clear that \autoref{conj:generalconj} implies \autoref{conj: case w = 3}. In order to prove the other implication, we prove two preliminary results showing that one can reduce to pairs in $\bfU_w$ whenever $w \geq 4$.
\begin{lemma}\label{lem:vwequtow}
Let $w \geq 4$ and let $(a,b)\in \bfV_w$. Then $(a,b)$ is castling equivalent to a pair $(a_m,b_m) \in \bfU_w$. 
\end{lemma}
\begin{proof}
Define $(\tilde{a},\tilde{b}) = (a,aw-b)$, which is castling equivalent to $(a,b)$. Since $(a,b) \in \bfV_w$, we have
\begin{align*}
&    \tilde{b} = aw - b \geq aw - (w-1)a = a = \tilde{a}, \\
&    (w-1) \tilde{a} = (w-1)a = wa - a > wa - b = \tilde{b}.
\end{align*}
This guarantees $(\tilde{a},\tilde{b}) \in \bfU_w \cup \bfV_w$.

Since $\lambda_{w} \in (w-1,w)$ for $w \geq 3$, we have  
\begin{align*}
    \lambda_{w-1} > w-2 = \frac{w-4}{2} + \frac{w}{2} \geq \frac{w}{2}
\end{align*}
and therefore 
\begin{equation*}
    \tilde{b} = aw - b = aw - 2b + b \leq aw - 2\lambda_{w-1}a + b = 2a (\frac{w}{2} - \lambda_{w-1})<b.
\end{equation*}
Recursively define $(a_0,b_0) = (a,b)$ and $(a_{i},b_{i}) = (\tilde{a}_{i-1},\tilde{b}_{i-1})$. This defines a sequence of castling equivalent pairs $(a_i,b_i) \in \bfU_w \cup \bfV_w$. Moreover, the sequence $b_j$ is strictly decreasing, whereas $a_j = a$ for every $j$. Therefore, there exists $m \geq 0$ such that $b_m < \lambda_{w-1} a_m$, namely $(a_m,b_m) \in \bfU_w$. 
\end{proof}

Using \autoref{lem:vwequtow}, one can reduce the computation of the max-flow for pairs in $\bfV_w$ and $\bfW_w$ to pairs in $\bfU_w$. This is recorded in the following result:
\begin{lemma}\label{lemma: full rank from u to vw}
Let $w \geq 4$. If for all $(a,b),(a',b') \in \bfU_w$
\begin{equation}\label{eq: full rank from u to vw}
    \qmincut \left( \bridgegraphsmall{a}{b}{w}{b'}{a'} \right) = \qmaxflow \left( \bridgegraphsmall{a}{b}{w}{b'}{a'} \right) = \min \lbrace a' b, ab'\rbrace
\end{equation}
then \eqref{eq: full rank from u to vw} the same holds for $(a,b), (a',b') \in \bfV_w \cup \bfW_w$.
\end{lemma}
\begin{proof}
Let  $(a,b), (a',b') \in \bfV_w \cup \bfW_w$. If one of them is in $\bfV_w$ and one of them in $\bfW_w$, the result follows from \autoref{thm:wxvsuv}. 

Therefore, we may assume $(a,b),(a',b') \in \bfV_w$ or $(a,b),(a',b') \in \bfW_w$ and by \autoref{lem:behaviorundercastling}, case (c), we can reduce to the case where $(a,b),(a',b') \in \bfW_w$. If the two pairs have different castling depth, the result follows from \autoref{thm:diffcastlingdegree}. 

If $(a,b),(a',b')$ have the same castling depth, apply exactly $\depth(a,b)$ castling steps and obtain two pairs in $\bfU_w \cup \bfV_w$. If one of them is in $\bfU_w$ and one in $\bfV_w$, the result follows by \autoref{cor:v versus u}. If they both belong to $\bfU_w$, the result follows from the hypothesis. If they both belong to $\bfV_w$, apply the reduction of \autoref{lem:vwequtow} until (at least) one of the two pairs belongs to $\bfU_w$. This concludes the proof, either by \autoref{cor:v versus u} or by assumption.
\end{proof}

We can now easily complete the proof of the equivalence between \autoref{conj: case w = 3} and \autoref{conj:generalconj}:

\begin{proposition}\label{prop:conjs equ}
If 
\begin{equation*}
\qmincut \left( \bridgegraphsmall{a}{b}{3}{b'}{a'} \right)
= \qmaxflow \left( \bridgegraphsmall{a}{b}{3}{b'}{a'} \right) 
\end{equation*}
 for all $(a,b),(a',b') \in \bfU_3 \cup \bfV_3 \cup \bfW_3$ then 
\begin{equation*}
\qmincut \left( \bridgegraphsmall{a}{b}{w}{b'}{a'} \right) 
= \qmaxflow \left( \bridgegraphsmall{a}{b}{w}{b'}{a'} \right) 
\end{equation*}
for all $(a,b),(a',b') \in \bfU_w \cup \bfV_w \cup \bfW_w$. 
\end{proposition}
\begin{proof}
We use induction on $w$. The case $w=3$ coincides with the hypothesis of the claim. Suppose $w \geq 4$. Let $(a,b),(a',b') \in \bfU_w \cup \bfV_w \cup \bfW_w$. By \autoref{lemma: full rank from u to vw}, it suffices to assume $(a,b),(a',b') \in \bfU_w$.

Now, $\bfU_w = \bfU_{w-1} \cup \bfV_{w-1} \cup \bfW_{w-1}$. Therefore the statement holds by the induction hypothesis and this concludes the proof.
\end{proof}


\subsection*{Conflict of Interest} The authors have no conflict of interest to declare that is relevant to this article.

{
\bibliographystyle{alphaurl}
\bibliography{BridgeGraph.bib}
}

\begin{appendix}
\section{Pedestrian-style proofs for the bridge 2 case}
When $w = 2$, the regions $\bfU_2,\bfV_2$ are empty and the region $\bfW_2$ reduces to cases where $a=b$. Therefore, the only interesting behaviour in this case is for $(a,b),(a',b') \in \bfX_2$; this is covered by \autoref{thm:xvswx} whose proof relies on the castling transform and on \autoref{thm:QMaxFlowcastling}. In this Appendix, we propose two alternative proofs for this specific case, which do not rely on castling transform. A deep understanding of this initial case might provide insights on how to obtain a proof for the $w=3$ case of \autoref{conj: case w = 3} and in turn of \autoref{conj:generalconj}. 

The first proof that we present relies on the fact that when $w = 2$, we know explicit examples of tensors $T \in \bbC^a \otimes \bbC^b \otimes \bbC^2$ having $(\GL_a \times \GL_b)$-orbit of maximal dimension. This relies on the Kronecker classification of matrix pencils \cite[Ch.XIII]{Gant:TheoryOfMatrices}, and more precisely on the results of \cite{Pok:PerturbationsEquivalenceOrbitMatrixPencil}.

The second proof is an enhancement of the first one, where the explicit calculation of the rank of certain maps relies on the representation theory of $\SL_2$. 

\subsection{A proof using explicit tensors with dense orbit}\label{app:explicitorbit}

The main tool for this section is the following result by Pokrzywa~\cite{Pok:PerturbationsEquivalenceOrbitMatrixPencil}. Recall that $z^{(2)}_p =p$.

\begin{lemma}\label{lemma: dense orbit 2 a b}
Let $(a,b) \in \bfX_2$, namely $a < b \leq 2a$. Write 
\begin{align*}
a &= p \alpha + (p+1) \beta \\ 
b &= (p+1) \alpha + (p+2) \beta  
\end{align*}
as in \autoref{lemma: arithmetic a b p}. Let $T \in \mathbb{C}^a \otimes \mathbb{C}^b \otimes \mathbb{C}^2$ be regarded as a $2$-dimensional subspace of $\mathbb{C}^a \otimes \mathbb{C}^b \simeq \text{Mat}_{a \times b}$ as 
\[
 T((\mathbb{C}^{2})^*) = \lbrace (\mathbf{I}_\alpha \boxtimes R_p (\xi_1,\xi_2)  \oplus ( \mathbf{I}_\beta \boxtimes R_{p+1}(\xi_1,\xi_2)): \xi_1,\xi_2 \in \mathbb{C}\rbrace
\]
where 
\[
R_p (\xi_1,\xi_2) = \left[ 
\begin{array}{ccccc}
\xi_1 & \xi_2 & & & \\    
& \xi_1 & \xi_2  & & \\ 
& & \ddots & \ddots & \\
& & & \xi_1 & \xi_2 
\end{array}
\right] \in \mathbb{C}^{p} \otimes \mathbb{C}^{p+1}
\]
and $\textbf{I}_\alpha$ is an $\alpha \times \alpha$ identity matrix. Then the $(GL_a \times GL_b)$-orbit of $T$ is dense in $\mathbb{C}^a \otimes \mathbb{C}^b \otimes \mathbb{C}^2$.
\end{lemma}

The proof of \autoref{lemma: dense orbit 2 a b} can be easily obtained by computing directly the stabilizer of the tensor $T$ in $\GL_a \times \GL_b$. The direct sum structure allows one to reduce to the case $\alpha = 0$, $\beta = 1$, for which the calculation is straightforward; see e.g. \cite[Theorem 4.1]{ConGesLanVenWan:GeometryStrassenAsyRankConj}.

\autoref{lemma: dense orbit 2 a b}, together with semicontinuity of matrix rank, guarantees that the quantum max-flow in the bridge graph equals the rank of the flow map for tensors $T,T'$ having dense orbit. The rest of this section performs the calculation of this rank. 

Fix $(a,b),(a',b') \in \bfX_2$ and write 
\begin{equation*}
  \begin{array}{lccl}
	 a = p \alpha + _(p+1) \beta, & &  &b = ({p+1}) \alpha + ({p+2}) \beta \\
	 a' = {p'} \alpha' + ({p'}+1) \beta', & & &b' = ({p'}+1) \alpha' + ({p'}+2) \beta',
  \end{array}
\end{equation*}
as in \autoref{lemma: arithmetic a b p}. Let $T \in \mathbb{C}^{a}\otimes \mathbb{C}^{b}\otimes \mathbb{C}^{2}$ and $T' \in \mathbb{C}^{b'}\otimes \mathbb{C}^{a'}\otimes \mathbb{C}^{2}$ be the tensors described in \autoref{lemma: dense orbit 2 a b}.

Fix a basis $e_1,\dots ,e_k$ of $\mathbb{C}^{k}$ for any $k$. For $j = 1,2$, let $M_j = T(e_j^*) \in \mathbb{C}^{a}\otimes \mathbb{C}^{b}$ and $M_j' = T'(e_j^*) \in \mathbb{C}^{b'} \otimes \mathbb{C}^{a'}$. By \autoref{lem:reformulationof qmaxflow}, we have
\[
 F_{T,T'} = M_1 \boxtimes M_1'+ M_2 \boxtimes M_2'.
\]
A direct calculation shows 
\begin{align*}
M_1 &= I_\alpha \boxtimes R_p(1,0) \oplus I_\beta \boxtimes R_{p+1}(1,0) , \\
M_2 &= I_\alpha \boxtimes R_p(0,1) \oplus I_\beta \boxtimes R_{p+1}(0,1) , \\
M_1' &= I_\alpha \boxtimes R_{q}(1,0)^T \oplus I_\beta \boxtimes R_{q+1}(1,0)^T , \\
M_2' &= I_\alpha \boxtimes R_{q}(0,1)^T \oplus I_\beta \boxtimes R_{q+1}(0,1)^T. 
\end{align*}

After reordering the Kronecker factors, we obtain that $F_{T,T'}$ can be represented in diagonal block form with 
\[
 \begin{array}{lc}
  \alpha \alpha' \text{ blocks} & R_p(1,0) \boxtimes R_q(1,0)^T  + R_p(0,1) \boxtimes R_q(0,1)^T, \\
  \alpha \beta' \text{ blocks} & R_p(1,0) \boxtimes R_{q+1}(1,0)^T  + R_p(0,1) \boxtimes R_{q+1}(0,1)^T, \\
  \beta \alpha' \text{ blocks} & R_{p+1}(1,0) \boxtimes R_q(1,0)^T  + R_{p+1}(0,1) \boxtimes R_q(0,1)^T, \\
  \beta \beta' \text{ blocks} & R_{p+1}(1,0) \boxtimes R_{q+1}(1,0)^T  + R_{p+1}(0,1) \boxtimes R_{q+1}(0,1)^T.
 \end{array}
\]
Consequently, the rank of $F_{T,T'}$ coincides with the sum of the ranks of these blocks.

Define
\[
K_{x, y} = R_x(1,0) \boxtimes R_y(1,0)^T  + R_x(0,1) \boxtimes R_y(0,1)^T \in (\mathbb{C}^x \otimes \mathbb{C}^{y + 1}) \otimes (\mathbb{C}^{x + 1} \otimes \mathbb{C}^y).
\]
The four block types appearing in $F_{T,T'}$ are $K_{p,p'}$, $K_{p, p'+ 1}$, $K_{p + 1, p'}$ and $K_{p+1, p'+1}$.

Consider $K_{x, y}$ as a linear map $K_{x, y} \colon\mathbb{C}^x \otimes \mathbb{C}^{y + 1} \to \mathbb{C}^{x+1} \otimes \mathbb{C}^y$.
On the basis vector it acts as follows:
\[
 K_{x, y}(e_i \otimes e_j) = \begin{cases}
 e_i \otimes e_j & j = 1 \\
	 e_i \otimes e_j + e_{i+1} \otimes e_{j-1} , & 2 \leq j \leq y \\
 e_{i+1} \otimes e_{j -1} & j = y + 1.
 \end{cases}
\]

Denote by $U_c$ the subspace of $\mathbb{C}^{x} \otimes \mathbb{C}^{y+1}$ spanned by the basis vectors $e_i \otimes e_j$ with $i + j = c + 1$, and by $V_c$ the analogous subspace of $\mathbb{C}^{x + 1} \otimes \mathbb{C}^y$.
The map $K_{x, y}$ maps $U_c$ to $V_c$.
The matrix of $K_{x , y}$ restricted to $U_c$ has a simple structure with ones on two diagonals.
Depending on the relation between $x$, $y$ and $c$ we have the following $4$ cases.
\[
\begin{cases}
\rank K_{x, y}|_{U_c} = \dim U_c = \dim V_c = c, & \text{if $c \leq x$, $c \leq y$}, \\
\rank K_{x, y}|_{U_c} = \dim U_c = x,\ \dim V_c = y + 1, & \text{if $x < c \leq y$}, \\
\rank K_{x, y}|_{U_c} = \dim V_c = y,\ \dim U_c = y + 1, & \text{if $y < c \leq x$}, \\
\rank K_{x, y}|_{U_c} = \dim U_c = \dim V_c = x + y + 1 - c, & \text{if $c > x$, $c > y$}. \\
\end{cases}
\]
If $x \leq y$, then $\rank (K_{x, y}|_{U_c}) = \dim U_c$ for all $c$ and, therefore, $\rank K_{x, y} = \sum_c \dim U_c = x(y + 1)$. Similarly, if $x \geq y$, then $\rank K_{x, y} = y(x + 1)$.

With these considerations, we can now compute $\rank(F_{T,T'})$. Clearly if $p < q$ then $ p+1 \leq q$. Consequently, all appearing maps $K_{x,y}$ have rank $x(y+1)$ and we obtain
\begin{align*}
	\rank F_{T,T'} &= \alpha \alpha' \rank K_{p,p'} +  
	\alpha \beta' \rank K_{p,p'+1} +  
\beta \alpha' \rank K_{p+1,p'} +  
	\beta \beta' \rank K_{p+1,p'+1} \\
&= \alpha \alpha' p (p'+1) + \alpha \beta' p (p'+2) + 
\beta \alpha' (p+1) (p'+1) + 
\beta \beta' (p+1) (p'+2) = ab'. 
\end{align*}
This shows that if $p < p'$, then the quantum max-flow equals the quantum min-cut, and both are $ab'$. Similarly, if $p > p'$, the quantum max-flow and quantum min-cut coincide and are equal to $a'b$.

If $p = q$, we calculate the rank of $F_{T,T'}$ as 
\begin{align*}
	\rank F_{T,T'} &= \alpha \alpha' \rank K_{p,p'} +  
	\alpha \beta' \rank K_{p,p'+1} +  
\beta \alpha' \rank K_{p+1,p'} +  
	\beta \beta' \rank K_{p+1,p'+1} \\
&= \alpha \alpha' p (p'+1) + \alpha \beta' p(p'+2) + 
\beta \alpha' (p+2)p' + 
\beta \beta' (p+1) (p'+2)
\end{align*}

This quantity differs from 
\begin{equation*}
ab' =  \alpha \alpha' p (p'+1) + \alpha \beta' p(p'+2) +
\beta \alpha' (p+1)(p'+1) + 
\beta \beta' (p+1) (p'+2)
\end{equation*}
by $(p - p' + 1) \beta \alpha' = \beta \alpha'$. Hence, we have reproduced the result in \autoref{thm:xvswx}
\begin{equation*}
\qmaxflow \left( \bridgegraphsmall{a}{b}{2}{b'}{a'} \right) = 
\begin{cases}
	ab' & \text{if } p < p',\\
	ab'-\beta \alpha' & \text{if } p = p', \\
	a'b & \text{if } p > p'.
\end{cases}
\end{equation*}


\subsection{Bridge 2 via representation theory }\label{app:rep theory}

In this section, we demonstrate how one can use the representation theory of $\GL_2$ to perform the rank calculation of the \autoref{app:explicitorbit}.
We will also see that this idea lets us calculate quantum max-flows for higher bridge dimensions in particular cases.

First, we recall some basic notions from representation theory. We refer to \cite{FulHar:RepTh} for an in-depth introduction. Throughout this section, let $V$ be a complex finite dimensional vector space with $\bfv = \dim V$ and fix a basis $e_1 \dots e_\bfv$ of $V$. The \textit{symmetric subspace} $S^d V \subset V^{\otimes d}$ is the subspace spanned by vectors 
\begin{equation*}
	e_{i_1}\dots e_{i_d} = \sum_{\sigma \in S_d} e_{\sigma(i_1)} \otimes \dots \otimes e_{\sigma(i_d)}
\end{equation*}
where $S_d$ is the permutation group on $d$ symbols. For all $1 \leq i_1 \leq \dots \leq i_d \leq \bfv $, the vectors $e_{i_1}\dots e_{i_d}$ are linearly independent and consequently, 
\begin{equation*}
	\dim (S^dV ) = \binom{\bfv + d - 1}{d}.
\end{equation*}
The \textit{antisymmetric subspace} $\Lambda^dV \subset V^{\otimes d}$ is spanned by elements of the form 
\begin{equation*}
	e_{i_1}\wedge \dots\wedge e_{i_d} = \sum_{\sigma \in S_d} \sgn (\sigma)  e_{\sigma(i_1)} \otimes \dots \otimes e_{\sigma(i_d)}.
\end{equation*}
Denote by $x_1 \vvirg x_{\bfv}$ the basis of $V^*$ dual to $e _{1} \dots e _\bfv$. The space $S^d V^*$ can be identified with the space of homogeneous polynomials of degree $d$ on $V$; the basis elements $x_1 \vvirg x_\bfv$ can be thought as coordinates on $V$. In a similar way, $S^kV$ can be thought of as the space of order $k$ differential operators with constant coefficients, via the natural contraction map
\begin{align*}
	S^d V^* \otimes S^k V &\rightarrow S^{d-k}V^* \\
	f \otimes e_{i_1 } \dots e_{i_k} &\mapsto \frac{\partial^k f}{\partial x_{i_1} \cdots \partial x_{i_k}}
\end{align*}
for any $d \geq k$. 

The group $\GL(V)$ acts on $V^{\otimes d}$ and its elements can be represented as matrices in the fixed basis. Let $\bfT$ be the abelian subgroup of invertible diagonal matrices (called a \textit{torus}) and $\bfB$ be the subgroup of invertible upper triangular matrices (called the \textit{Borel subgroup}). A nonzero tensor $v \in V^{\otimes d}$ is a \textit{weight vector} with \textit{weight} $(p_1 \dots p_{\bfv})$ if for all elements in $\textbf{T}$, one has 
\begin{equation*}
\begin{pmatrix}
	t_1 & & \\
	    &\ddots & \\
	        &&t_{\bfv}
\end{pmatrix}v = t_1^{p_1} \dots t_{\textbf{v}}^{p_{\textbf{v}}} v.
\end{equation*}
Moreover, $v$ is a \textit{highest weight vector} if $\bfB$ preserves the line through $v$; in this case the weight $(p_1 \dots p_{\bfv})$ is called a \textit{highest weight} .

Two irreducible representations of $\GL(V)$ are isomorphic if and only if they have the same highest weight, that is, contain a highest weight vector for the same highest weight. The irreducible representations appearing in the tensor algebra $V^{\otimes } = \bigoplus_{d = 0}^{\infty} V^{\otimes d}$ are labeled by Young diagrams or, equivalently, by tuples of natural numbers $\pi = (p_1 \dots p_k)$ with $p_i \leq p_{i + 1}$. Write $\bbS_{\pi} V$ for the irreducible representation corresponding to $\pi$. One can show that $\bbS_\pi V$ has highest weight $\pi$; moreover, $\bbS_\pi V$ is an irreducible representation appearing in $V^{\otimes d}$ if and only if $p_1 + \dots +  p_k = d$, with $k \leq \bfv$. Let $\pi' = (q_1 \dots q_{p_1})$ be the conjugate Young diagram to $\pi$, that is the Young diagram whose $i$-th column has $p_i$ boxes. Then, 
\begin{equation*}
	(e_1 \wedge \dots \wedge e_{q_1})\otimes \dots \otimes (e_1 \wedge \dots \wedge e_{q_{p_1}}) \in V^{\otimes d}
\end{equation*}
is a highest weight vector of highest weight $\pi$ in $V^{\otimes d}$. 

Every finite-dimensional representation of $\GL(V)$ splits into a direct sum of irreducible representations. The Pieri formula controls which irreducible representations appear in $\bbS_\pi V \otimes S^dV$. 

\begin{theorem}\label{thm:pieriformula}
The representation $\bbS_\pi V \otimes S^dV$ decomposes into irreducible representations of $\GL(V)$ as 
\begin{equation*}
	\bbS_\pi V \otimes S^dV = \bigoplus_{\mu} \bbS_{\mu}V
\end{equation*}
where the direct sum ranges over Young diagrams $\mu$ obtained by adding $d$ boxes to $\pi$ with no two boxes added to the same column. In particular, this decomposition is multiplicity free.
\end{theorem}

For example, let $W$ be a $2$-dimensional space and consider the action of $\mathrm{GL}(W^*)$ on the space $S^k( W^*) \otimes S^m (W^*)$ with $k \geq m$. By Pieri's formula, this decomposes as  
\begin{equation*}
	S^{k}(W^*) \otimes S^{m}(W^*) = \bigoplus_{r = 0}^m \mathbb{S}_{(k+r, m-r)} (W^*),
\end{equation*}
because the Young diagram $(k+r, m-r)$ is obtained by adding $m$ boxes to the Young diagram $(k)$, with $r$ boxes in the first row and $m-r$ boxes in the second row.

Define $h_{k,m}^r = (x_0 y_1 - x_1 y_0)^{m - r} x_0^{k-m  + r} y_0^{r}$. Identify elements of $S^{k}(W^{*}) \otimes S^m (W^*)$ as polynomials in the variables $x_0,x_1,y_0,y_1$ which are bi-homogeneous of bi-degree $(k,m)$ in the two sets of variables $\{x_0,x_1\},\{y_0,y_1\}$. In particular, $h_{k,m}^r$ is an element of $S^{k}(W^*) \otimes S^{m}(W^*)$. One can verify that this is, in fact, a highest weight vector for highest weight $(k + r, m - r)$. 

Similarly, for $k \leq m$, the Pieri formula gives 
\begin{equation*}
	S^{k}(W^*) \otimes S^{m}(W^*) = \bigoplus_{r = 0}^m \mathbb{S}_{(m + r, k - r)} (W^*).
\end{equation*}
Define $h'^r_{k,m} = (x_0 y_1 - x_1 y_0)^{k - r} x_0^{m - k + r} y_0^{r}$, which turns out to be a highest weight vector in $\mathbb{S}_{(m +r,k - r)}(W^*)$ in $S^{k}(W^*) \otimes S^{m}(W^*)$. 

In the study of the bridge graph, with $w=2$, we may interpret the spaces $A,B,A',B'$ as representations for the group $\mathrm{GL}(W^*)$. Suppose preliminarily $\dim (A) = p+1$, $\dim (B) = p+2$, $\dim (A') = p'+1$ and $\dim (B') = p'+2$, and consider 
\begin{equation*}
	A \cong S^p(W),\; B \cong S^{p+1}(W^*),\; A' \cong S^{p'}(W^*),\; B' \cong S^{p'+1}(W).
\end{equation*}

By Schur's lemma, there is an (up to scaling) unique $\mathrm{GL}(W^*)$-equivariant projection 
\[
S^{p} (W^*) \otimes W^* \to S^{p+1} (W^*).
\]
 This corresponds to a tensor $T \in A \otimes B \otimes W$ that can be interpreted as the multiplication map $T(\ell \otimes g) = \ell g$.

By Pieri's formula, $W^* \otimes S^{p'}(W^*)$ contains a copy of $S^{p'+1}(W^*)$. By Schur's Lemma, there is a unique, up to scaling, $\mathrm{GL}(W^*)$-equivariant embedding $S^{p'+1}(W^*)  \to  W^* \otimes S^{p'} (W^*)$. Let $T' \in A' \otimes B' \otimes W^*$ be the tensor corresponding to this embedding, regarded as an element of $S^{p'}(W^*) \otimes S^{p'+1} (W) \otimes W^*$. This map is called \textit{polarization} in~\cite[Section~2.6.4]{Lan:TensorBook} and is explicitly given by $T' (f)  = x_0 \otimes \frac{\partial }{\partial x_0} f +  x_1 \otimes \frac{\partial }{\partial x_1}f$.

Note that in the tensors $T$ and $T'$ are isomorphic to the tensors $R_p$ and $R_{p'}$ from~\autoref{lemma: dense orbit 2 a b}, which can be seen by writing them in the monomial basis. We will now show that the explicit rank computation in~\autoref{lemma: dense orbit 2 a b} can be replaced by a representation-theoretic argument.

The flow map $F_{T,T'}$ in the bridge graph with this choice of $T$ and $T'$ is  
\begin{equation}\label{eq:appendixFTT'}
\begin{aligned}
	F_{T,T'} : S^p(W^*) \otimes S^{p'+1}(W^*) &\rightarrow S^{p+1}(W^*) \otimes S^{p'}(W^*) \\
	f \otimes g &\mapsto {\textstyle \sum}_{i = 0,1} f \cdot x_i \cdot \frac{\partial}{\partial y_i} g
\end{aligned}
\end{equation}
where, again, the tensor product of two symmetric powers is interpreted as a space of bi-homogeneous polynomials. 

Using Pieri's formula and Schur's Lemma, we compute the rank of this map for every choice of $p$ and $p'$. 
\begin{lemma}\label{lem:appendixpieribridgeformula}
The flow map $F_{T,T'}$ from~\autoref{eq:appendixFTT'} is $\mathrm{GL}(W^*)$-equivariant. Moreover, the map is injective if $p \leq p'$ and surjective if $p \geq p'$. In particular, $F_{T,T'}$ is an isomorphism if and only if $p= p'$. 
\end{lemma}
\begin{proof}
Since the maps defined by $T$ and $T'$ are $\mathrm{GL}(W^*)$-equivariant, the flow map $F_{T,T'}$ is as well. Suppose $p \leq p'$. By Pieri's formula, the domain is 
\begin{equation*}
    S^p (W^*) \otimes S^{p'+1} (W^*) = \bigoplus_{r = 0}^p \mathbb{S}_{(p'+1+i,p-i)}(W^*).
\end{equation*}
Since the decomposition is multiplicity free, by Schur's Lemma, it suffices to check that the highest weight vectors of the domain are not mapped to $0$. The highest weight vectors of the domain are 
$h'^r _{p,p'+1} \in \mathbb{S}_{(p'+1+r,p-r)}(W^*) \subset S^p (W^*) \otimes S^{p'+1} (W^*)$.

It is easy to verify that $F_{T,T'}(h'^r_{p,p'+1})$ is a nonzero multiple of $h^r_{p+1,p'}$. This guarantees that $F_{T,T'}$ is injective. 

If $p \geq p'$, the same calculation can be performed to show that the transpose map $F_{T,T'}^{\mathbf{\textit{t}}}$ is injective. Hence $F_{T,T'}$ is surjective. 

We conclude that if $p=p'$, $F_{T,T'}$ is an isomorphism. On the other hand, if $p \neq p'$, domain and codomain do not have the same dimension, hence the map is not an isomorphism.
\end{proof}

To calculate the quantum max-flow for general bond dimensions $a,b,a'$ and $b'$, we use the same method as in~\autoref{app:explicitorbit}. Write 
\begin{equation*}
  \begin{array}{lccl}
	 a = p \alpha + (p+1) \beta, & &  &b = ({p+1}) \alpha + ({p+2}) \beta \\
	 a' = {p'} \alpha' + ({p'}+1) \beta', & & &b' = ({p'}+1) \alpha' + ({p'}+2) \beta',
  \end{array}
\end{equation*}
and assume without loss of generality $p \leq p'$. We can define tensors in $A \otimes B \otimes W$ resp.~$A' \otimes B' \otimes W^*$ as block tensors where the block elements are the distinguished tensors $T$ and $T'$ defined before. Now, using~\autoref{lem:appendixpieribridgeformula}, one can repeat the same calculation as in~\autoref{app:explicitorbit} to calculate the quantum max-flow.   

%
\end{appendix}

\end{document}